\documentclass{statsoc}

\usepackage[a4paper]{geometry}
\usepackage{amssymb, amsmath, multirow}
\usepackage{etoolbox}

\makeatletter
\patchcmd{\@makecaption}
  {\parbox}
  {\advance\@tempdima-\fontdimen2} % decrease the width!
  {}{}
\makeatother  

\usepackage{graphicx}
\usepackage[caption = false]{subfig} 
\usepackage{enumerate}
\usepackage{comment}
\usepackage{natbib, float}

\usepackage{epsfig, amsmath, amssymb, latexsym, url, xcolor, bbm}

\usepackage{float} 
\floatstyle{plain}
\newfloat{Algorithm}{thp}{lop}
\floatname{Algorithm}{Algorithm}
\def\bibfont{\fontsize{9}{10}\selectfont}
\RequirePackage[colorlinks,citecolor=blue,urlcolor=blue]{hyperref}
\newtheorem{theorem}{Theorem}
\newtheorem{corollary}{Corollary}
\newtheorem{definition}{Definition}
\newtheorem{lemma}{Lemma}

\def\tr{\text{\rm tr}}
\def\vec{\text{\rm vec}}
\def\d{\text{\rm d}}
\def\cov{\text{\rm cov}}
\def\var{\text{\rm var}}
\def\logit{\text{\rm logit}}
\def\KL{\text{\rm KL}}
\def\ML{\text{\rm ML}}

\def\argmax{\text{\rm argmax}}
\def\diag{\text{\rm diag}}
\def\dg{\text{\rm dg}}
\def\vech{v}
\def\hnabla{\widehat{\nabla}}
\def\heta{\hat{\eta}}
\def\hg{\hat{g}}
\def\hH{\hat{H}}
\def\tilb{\tilde{b}}
\def\tiltheta{\tilde{\theta}}

\def\L{{\mathcal{L}}}
\def\N{{\mathcal{N}}}

\graphicspath{{img/}}

\title[Reparametrized variational Bayes]{Use of model reparametrization to improve\\
variational Bayes}

\author[L. S. L. Tan]{Linda. S. L. Tan\footnote{Linda Tan was supported by the start-up grant R-155-000-190-133.}}

\address{National University of Singapore, Singapore}

\email{statsll@nus.edu.sg}

\begin{document}

\begin{abstract}
We propose using model reparametrization to improve variational Bayes inference for hierarchical models whose variables can be classified as global (shared across observations) or local (observation specific). Posterior dependence between local and global variables is minimized by applying an invertible affine transformation on the local variables. The functional form of this transformation is deduced by approximating the posterior distribution of each local variable conditional on the global variables by a Gaussian density via a second order Taylor expansion. Variational Bayes inference for the reparametrized model is then obtained using stochastic approximation. Our approach can be readily extended to large datasets via a divide and recombine strategy. Using generalized linear mixed models, we demonstrate that reparametrized variational Bayes (RVB) provides improvements in both accuracy and convergence rate compared to state of the art Gaussian variational approximation methods.
\end{abstract}


\section{Introduction} \label{sec:Introduction}
Hierarchical models account for the statistical dependence of observations within clusters, by organizing them in a multilevel structure and introducing parameters which are cluster-specific or shared across clusters. Such models are immensely flexible and have found applications in many fields, ranging from social science \citep{Gelman2006}, medicine \citep{McCormick2012} to ecology \citep{Cressie2009}. The linear mixed model (LMM) and generalized linear mixed model (GLMM), for instance, are highly popular in epidemiological and biomedical studies where observations are often correlated and hierarchical in nature \citep{Liu2016}. However, hierarchical modeling poses computational challenges for practitioners, whether a Bayesian or maximum likelihood approach is used for inference. Here we consider a Bayesian framework, which accounts for uncertainty and allows for the incorporation of prior information from previous studies. A variety of hierarchical models can be fitted in the Bayesian framework using Markov chain Monte Carlo (MCMC) methods via standard software such as BUGS \citep{Lunn2009} or Stan \citep{Carpenter2017}. Alternatively, fast approximate Bayesian inference can be obtained for latent Gaussian models by using integrated nested Laplace approximation \citep[INLA,][]{Rue2009}, and for a broader class of models by using variational approximation \citep{Jacobs1991, Blei2017}. 

The parametrization of a hierarchical model has a huge impact on the performance of the statistical method used for fitting it, and two popular approaches are {\em centering} and {\em noncentering}. As illustration, the normal hierarchical model, $y|x \sim N(x,\sigma_y^2)$ and $x|\theta \sim N(\theta, \sigma_x^2)$, where $y$ is observed and $\theta$ is the unknown parameter, is {\em centered} as the latent variable $x$ is centered about $\theta$. The equivalent model, written as $y|\tilde{x}, \theta \sim N(\tilde{x} + \theta,\sigma_y^2)$ and $\tilde{x} \sim N(0, \sigma_x^2)$ is {\em noncentered} as $\tilde{x}$ is independent of $\theta$ a priori. \cite{Gelfand1995, Gelfand1996} observe that MCMC samplers for fitting LMMs and GLMMs converge slowly when there is weak identifiability or high correlation among variables, and reparametrizations in the form of hierarchical centering improve the rate of convergence greatly. However, \cite{Papaspiliopoulos2003, Papaspiliopoulos2007} show that centering performs well only if the data are highly informative about the latent process, and noncentering is preferred if the latent process is only weakly identified. For the normal hierarchical model, $y$ is informative about $x$ if $\sigma_y^2 \ll \sigma_x^2$ (centering is preferred), and $x$ is weakly identified if $\sigma_y^2 \gg \sigma_x^2$ (noncentering performs better). In fact, centering and noncentering are complementary as the Gibbs sampler will converge faster under one parametrization if it converges slowly under the other. For the best of both worlds, \cite{Yu2011} introduce an ancillarity-sufficiency strategy that interweaves the centered and noncentered parametrizations to improve the efficiency of MCMC algorithms. 

While centering and noncentering are motivated by improving the prior geometry of latent variables to minimize ``missing" information, reparametrizations that target the posterior geometry may yield even better results. In the MCMC context, \cite{Papaspiliopoulos2003, Papaspiliopoulos2007} introduce a {\em partially noncentered} parametrization for local variables in LMMs so that they are a posteriori independent of each other and the global variables. This parametrization lies on the continuum between centering and noncentering and selects automatically the best parametrization. \cite{Christensen2006} propose data-based reparametrizations for spatial GLMMs that reduce posterior correlation among variables so that mixing and convergence of MCMC algorithms are more robust. INLA, which focuses on accurate marginal posterior approximations of latent Gaussian models, also uses a reparametrization of the global parameters (based on their posterior mode and covariance) to correct for scale and rotation, thus aiding exploration of the posterior marginal and simplifying numerical integration. In this article, we propose a reparametrization of the local variables that improves variational Bayes \citep[VB,][]{Attias1999} inference for hierarchical models, by applying an invertible affine transformation to minimize posterior correlation among local and global variables.

 \subsection{Variational Bayes inference}
Variational approximation has become an increasingly popular alternative to MCMC methods for estimating posterior densities due to their ability to scale up to high-dimensional data. Suppose $y$ is the observed data and $\theta$ is the set of variables in a model. In variational approximation, some restriction is placed on the density $q(\theta)$ approximating the true posterior so that it is more tractable. The Kullback-Leibler (KL) divergence between $q(\theta)$ and $p(\theta|y)$, $D_\KL(q||p) = \int q(\theta) \log \{q(\theta)/p(\theta|y)\} d\theta$, is then minimized subject to these restrictions \citep{Ormerod2010}. Common restrictions assume a parametric density such as a Gaussian or a product density, where $q(\theta) = \prod_{i=1}^n q_i(\theta_i)$ for some partition $\theta=\{\theta_1, \dots, \theta_n\}$. The latter, known as VB, has many advantages such as being low-dimensional, quick to converge, having closed form updates (for conditionally conjugate models) and scalability to large data \citep{Hoffman2013}. However, the resulting approximation can be poor if strong posterior dependencies exist among $\{\theta_i\}$. 

\cite{Tan2013} demonstrate that partial noncentering can improve VB inference for GLMMs in terms of convergence rate as well as posterior approximation accuracy. Suppose $\beta$, $b_i$ and $\Omega$ denote the fixed effects, random effects and random effects precision matrix for subject $i =1, \dots, n$ in a GLMM, and there are random effects corresponding to each fixed effect. In VB, we may consider $q(\theta) = q(\Omega)q(\beta)\prod_{i=1}^n q(b_i)$, where $\theta=\{\Omega, \beta, b\}$ and $b = [b_1^T, \dots, b_n^T]^T$. However, this approximation fails to capture the usually strong posterior dependence among $\{\beta, b\}$. If we transform $b_i$ such that $\tilb_i = b_i - w_i \beta$, then $w_i$ can be tuned so that posterior dependence between $\tilb_i$ and $\beta$ is minimized. An approximation $q(\tiltheta) = q(\Omega)q(\beta)\prod_{i=1}^n q(\tilb_i)$, where $\tiltheta = \{\Omega, \beta, \tilb_1, \dots, \tilb_n\}$ will then reflect the dependence structure in $p(\tiltheta|y)$ more accurately. However, this transformation does not account for posterior dependence between $b_i$ and $\Omega$, and difficulties remain in estimating the posteriors of $\Omega$ and regression coefficients in $\beta$ which cannot be centered accurately. \cite{Lee2016} and \cite{Rohde2016} consider an alternative VB approximation for GLMMs which groups $\beta$ and $b$ together in a single factor $q(\beta, b)$, and makes use of inherent block-diagonal matrix structures to derive efficient updates.

We consider VB inference for hierarchical models where variables are classified as {\em global} (shared across observations) or {\em local} (observation specific). Our goal is to obtain a low-dimensional posterior approximation scalable to large data, by reparametrizing the local variables so that posterior dependence between local and global variables is minimized. We apply an invertible affine transformation that is a function of the global variables, and the functional form is deduced by considering a second order Taylor approximation to the posterior distribution of the local variables conditional on the global ones. This transformation is shown to be a generalization of the partially noncentered parametrization when location and scale parameters are unknown. We then consider independent Gaussian approximations to the posteriors of global and local variables. As closed form updates are not feasible, variational parameters are optimized using stochastic gradient ascent \citep{Titsias2014}. The assumption of posterior independence among transformed local and global variables allow our approach to be extended to large data easily using a divide and recombine strategy \citep{Broderick2013, Tran2016}. Application of this methodology is illustrated using GLMMs. However, our methods can be extended to a wider class of models including discrete choice models and spatial GLMMs. This approach of minimizing posterior dependence between global and local variables using model reparametrization prior to applying VB is referred to as {\em reparametrized variational Bayes} (RVB). 

Other approaches to relax the independence assumption in VB exist in the literature \citep[e.g.][]{Gershman2012, Salimans2013, Rezende2015}. \cite{Hoffman2015} develop structured stochastic variational inference for conditionally conjugate models, which allows local variables to depend explicitly on global variables in the variational posterior through some function $\gamma(\cdot)$ that is optimized numerically by maximizing a local lower bound. \cite{Titsias2017} propose model reparametrization for improving variational inference, and parameters of the affine transformation are optimized by minimizing the KL divergence between the variational approximation and a density obtained after a variable number of MCMC steps. Our approach is different as the functional form of our affine transformation is deduced from a Taylor approximation to the conditional posteriors of the local variables and remain fixed during optimization of the variational parameters. \cite{Kucukelbir2016} develop automatic differentiation variational inference in Stan, where the variational density can be a full Gaussian approximation. However, there may be difficulties in inferring such high-dimensional approximations. \cite{Tan2018} consider a Gaussian variational approximation (GVA) where posterior dependence is captured using sparse precision matrices.

Our article is organized as follows. Section \ref{sec_notation} introduces the notation, and Section \ref{sec_model} specifies the model and defines the affine transformation. Application of the affine transformation to GLMMs is illustrated in Section \ref{sec: app_GLMMs}. The stochastic variational algorithm is described in Section \ref{sec_variational} and Section \ref{sec_gradient} explains how stochastic gradients are computed. Extension of the variational algorithm to large datasets is discussed in Section \ref{sec:ext} and results for GLMMs are presented in Section \ref{sec_results}. Section \ref{sec_conclusion} concludes.

\section{Notation} \label{sec_notation}
For any $r \times r$ matrix $A$, let $\diag(A)$ denote the diagonal of $A$, $\dg(A)$ be the diagonal matrix derived from $A$ by setting non-diagonal elements to zero and $\bar{A}$ be the lower triangular matrix derived from $A$ by setting all superdiagonal elements to zero. Let $\text{vec}(A)$ denote the vector of length $r^2$ obtained by stacking the columns of $A$ under each other from left to right and $\vech(A)$ be the vector of length $r(r+1)/2$ obtained from $\vec(A)$ by eliminating all superdiagonal elements of $A$. Let $E_r$ denote the $r(r+1)/2 \times r^2$ elimination matrix such that $E_r \vec(A) = \vech(A)$ and $K_r$ be the $r^2 \times r^2$ commutation matrix such that $K_r \vec(A) = \vec(A^T)$. Let $N_r =  (K_r + I_{r^2})/2$. If $A$ is lower triangular, then $E_r^T \vech(A) = \vec(A)$. If $A$ is symmetric, then $D_r \vech(A) = \vec(A)$, where $D_r$ is the $r^2 \times r(r+1)/2$ duplication matrix. The Kronecker product between any two matrices is denoted by $\otimes$. Scalar functions applied to vector arguments are evaluated elementwise.

\section{Reparametrization using affine transformations} \label{sec_model}
 Suppose $y_i = [y_{i1},\dots,y_{in_i}]^T$ is the response vector and $b_i$ is the $r \times 1$ vector of latent variables for observation $i=1, \dots, n$. Let  $y=[y_1^T, \dots, y_n^T]^T$, $b=[b_1^T, \dots, b_n^T]^T$, $\theta_G$ be the $g \times 1$ vector of global variables, $\theta=[\theta_G^T, b^T]^T$ and $d= g + nr$ be the length of $\theta$. We focus on models whose joint density can be written as
\begin{equation} \label{mod}
p(y, \theta) = p(\theta_G) \prod_{i=1}^n p(b_i|\theta_G) p(y_i| b_i, \theta_G),
\end{equation}
where $p(\theta_G)$ is the prior of $\theta_G$, $p(b_i|\theta_G)$ is the density of $b_i$ conditional on $\theta_G$ and $p(y_i|b_i, \theta_G)$ is the likelihood of observing $y_i$ given $b_i$ and $\theta_G$. The posterior of $\theta$ has the structure $p(\theta|y) = p(\theta_G|y) \prod_{i=1}^n  p(b_i|\theta_G, y_i)$, where posteriors of the local variables are independent conditional on the global variables. Consider a variational approximation to $p(\theta|y)$ of the form
\begin{equation} \label{pdt1}
q(\theta) = q(\theta_G) \prod_{i=1}^n q(b_i).
\end{equation}
As $q(\theta)$ assumes posterior independence among $\{\theta_G, b_1, \dots, b_n\}$, this can be a poor choice if strong dependencies exist. To improve the variational approximation, we reparametrize the model by applying an invertible affine transformation to each local variable $b_i$, 
\begin{equation} \label{affine_transf}
\tilb_i = f(b_i|\theta_G) = L_i^{-1} (b_i - \lambda_i),
\end{equation}
where $\lambda_i$ (vector of length $r$) and $L_i$ ($r \times r$ lower triangular matrix with positive diagonal entries) are functions of $\theta_G$. The inverse transformation is $b_i =  f^{-1}(\tilb_i|\theta_G)  = L_i \tilb_i + \lambda_i$, and the functional forms of $\{\lambda_i, L_i\}$ will be deduced from $p(b_i|\theta_G, y_i)$. The motivation is as follows. Suppose $p(b_i|\theta_G, y_i)$ can be approximated by $N(\lambda_i, \Lambda_i) $, and $L_iL_i^T $ is the unique Cholesky decomposition of $\Lambda_i$. Then $\tilb_i|\theta_G, y \sim N(0, I_r)$ approximately, which implies that $\tilb_i$ is independent of $\theta_G$ in the posterior. Suppose we estimate the posterior of $\tiltheta = [\theta_G^T, \tilb^T]^T$ using
\begin{equation} \label{pdt2}
q(\tiltheta) = q(\theta_G) \prod_{i=1}^n q(\tilb_i),
\end{equation}
where $\tilb = [\tilb_1^T, \dots, \tilb_n^T]^T$. The product density assumption in \eqref{pdt2} is less restrictive than \eqref{pdt1} as the posterior dependence of each $\tilb_i$ on $\theta_G$ has been minimized. It reflects more accurately the dependency structure in $p(\tiltheta|y)$ and hence is expected to be more accurate.

\subsection{Generalization of the partially noncentered parametrization}
We use the LMM to illustrate how functional forms of $\{\lambda_i, L_i\}$ are deduced and show that the affine transformation is a generalization of the partially noncentered parametrization \citep{Tan2013} when variance components are unknown. For $i=1, \dots, n$, $j=1, \dots, n_i$, let $y_{ij} \sim N(\mu_{ij}, \sigma^2)$, where $\sigma^2$ is assumed known for simplicity,
\begin{equation*}
\mu_{ij} = X_{ij}^T \beta + Z_{ij}^T b_i, \quad\text{and} \quad b_i \sim N(0, \Omega^{-1}).
\end{equation*}
Here $\beta$ is a $p \times 1$ vector of fixed effects, $b_i$ is a $r\times 1$ vector of random effects and $X_{ij}$ and $Z_{ij}$ are covariates of length $p$ and $r$ respectively. Let $X_i = [X_{i1}, \dots, X_{in_i}]^T$ and $Z_i = [Z_{i1}, \dots, Z_{in_i}]^T$.  As $p(b_i|\beta, \Omega, y_i) \propto p(y_i|b_i, \beta, \Omega) p(b_i|\Omega)$, $b_i| \beta, \Omega, y_i \sim N(\lambda_i, \Lambda_i)$, where 
\begin{equation} \label{forms_LMM}
\Lambda_i^{-1} = \Omega + Z_i^T Z_i/\sigma^2, \quad
\lambda_i = \Lambda_i Z_i^T (y_i - X_i \beta)/\sigma^2.
\end{equation}
We seek a transformation $\tilb_i = f(b_i|\beta, \Omega)$ to minimize posterior dependence of $\tilb_i$ on $\{\beta, \Omega \}$. Suppose $\Omega$ is known and $X_i = Z_i$. Then $b_i$ depends on $\beta$ only through its mean and it suffices to consider 
\begin{equation} \label{transfmean}
\tilb_i = b_i + \Lambda_i X_i^T X_i \beta/\sigma^2, 
\end{equation}
where $L_i = I_r$ and $\lambda_i = - \Lambda_i X_i^T X_i \beta /\sigma^2$ in \eqref{affine_transf}. Then $\tilb_i$ is independent of $\beta$ a posteriori. Even if we assume $q(\tiltheta) = q(\beta) \prod_{i=1}^n q(\tilb_i)$, the optimal $q(\tiltheta)$ which minimizes $D_{\KL}(q(\tiltheta) || p(\tiltheta|y))$ is equal to $p(\tiltheta|y)$ as this product density structure is obeyed in $p(\tiltheta|y)$. 

\cite{Tan2013} introduce a partially noncentered parametrization,
\begin{equation*}
\tilb_i = b_i +(I_r - W_i) \beta,
\end{equation*}
where $W_i$ is a $r \times r$ tuning parameter. When $W_i = 0$, the parametrization is {\em centered} as $\tilb_i$ has a normal prior centered at $\beta$. When $W_i = I_r$,  $\tilb_i$ is independent of $\beta$ a priori and the parametrization is {\em noncentered}. The optimal value of $W_i$ is $\Lambda_i \Omega$, which leads to instant convergence and recovery of the true posterior in VB. This reparametrization is equivalent to \eqref{transfmean} as $I_r - W_i = \Lambda_i (\Lambda_i^{-1} - \Omega) = \Lambda_iX_i^T X_i/\sigma^2$. The centered and noncentered parametrizations correspond to $L_i=I_r$ and $\lambda_i = -\beta$ and $0$ respectively. Thus the optimal partially noncentered parametrization which minimizes $D_{\KL}(q(\tiltheta) || p(\tiltheta|y))$ actually transforms $b_i$ so that it is independent of $\beta$ in the posterior.

Extending this notion of minimizing posterior correlation to the case where $\Omega$ is unknown and $X_i \neq Z_i$, we can apply the transformation in \eqref{affine_transf} with $\lambda_i$ and $\Lambda_i$ given in \eqref{forms_LMM}. For LMMs, it is easy to identify $\lambda_i$ and $\Lambda_i$ as $p(b_i|\beta, \Omega, y_i)$ is Gaussian. More generally, we use Taylor expansions to obtain a Gaussian approximation to $p(b_i|\theta_G, y_i)$.

\section{Application to generalized linear mixed models} \label{sec: app_GLMMs}
For $i=1, \dots, n$, $j=1, \dots, n_i$, let $y_{ij}$ be generated from a distribution in the exponential family and $\mu_{ij} =  E(y_{ij})$ depend on the linear predictor, 
\begin{equation*}
\begin{aligned}
\eta_{ij} = X_{ij}^T \beta + Z_{ij}^T b_i,
\end{aligned}
\end{equation*}
through the link function $g(\cdot)$ such that $g(\mu_{ij}) = \eta_{ij}$. Here $\beta$, $b_i$, $X_{ij}$, $Z_{ij}$, $X_i$ and $Z_i$ are defined as for LMMs. We consider a diffuse prior $N(0, \sigma_\beta^2I_p)$ for $\beta$ where $\sigma_\beta^2$ is large, and the Wishart prior $W(\nu, S)$ for the precision matrix $\Omega$ where $\nu > r-1$ and $S$ is a $r \times r$ positive definite scale matrix. The hyperparameters $\nu$ and $S$ are determined based on the default conjugate prior proposed by \cite{Kass2006} and details are given in the supplementary material Section S1. Let $\eta_i = X_i \beta + Z_i b_i$ and $WW^T$ be the unique Cholesky decomposition of $\Omega$, where $W$ is lower triangular with positive diagonal entries. Define the $r \times r$ matrix $W^*$ such that $W_{ii}^* = \log(W_{ii})$, $W_{ij}^* = W_{ij}$ if $i \neq j$ and $\omega = \vech(W^*)$. The vector of global variables is thus $\theta_G = [\beta^T, \omega^T]^T$. 

We focus on GLMMs with canonical links, and the one-parameter exponential family whose natural parameter is $\eta_{ij} \in \mathbbm{R}$. Some examples are given in Table \ref{GLMM_defn}. 
\begin{table} 
\caption{\label{GLMM_defn} Poisson and binomial GLMMs with canonical links. $\text{Poisson}(\mu_{ij})$ denotes the Poisson distribution with mean $\mu_{ij}$ and $\text{Binomial}(m_{ij}, p_{ij})$ denotes the binomial distribution with $m_{ij}$ trials and success probability $p_{ij}$.} 
\centering
\fbox{
\begin{tabular}{lccc}
Model for $y_{ij}$ & $E(y_{ij})$ & Canonical link: $g(\mu_{ij})$ & $h_{ij}(\eta_{ij})$ \\
\hline
$\text{Poisson}(\mu_{ij})$ & $\mu_{ij}$ & $\log(\mu_{ij})$ & $\exp(\eta_{ij})$  \\
$\text{Binomial}(m_{ij}, p_{ij})$ & $m_{ij}p_{ij}$ & $\text{logit} (\mu_{ij}/m_{ij})$ & $m_{ij} \log\{ 1+ \exp(\eta_{ij}) \}$  
\end{tabular}}
\end{table}
The joint density is $p(y, \theta) = p(\beta) p(\omega) \prod\nolimits_{i=1}^n p(b_i|\omega) \prod_{j=1}^{n_i} p(y_{ij}|\eta_{ij})$
and 
\begin{equation*}
\begin{aligned}
\log p(y, \theta) &=\sum_{i=1}^n \bigg[ \sum_{j=1}^{n_i}  \{y_{ij} \eta_{ij} - h_{ij}(\eta_{ij})\} - \frac{1}{2}b_i^T \Omega b_i \bigg] + \frac{n}{2} \log|\Omega|  - \frac{\beta^T \beta}{2\sigma_\beta^2} + \log p(w) + C_0,
\end{aligned}
\end{equation*}
where $h_{ij}(\cdot)$ is the log partition function and $C_0$ is a constant independent of $\theta$. The induced prior $p(\omega) = p(\Omega) |\partial \vech(\Omega)/\partial \omega|$ and Section S2 of the supplementary material shows that
\begin{equation*}
\log p(\omega) = \log p(\Omega) + r \log 2 + \sum_{i=1}^r (r-i+2) \log(W_{ii}),
\end{equation*}
where $\log p(\Omega) = \tfrac{\nu-r-1}{2} \log|\Omega| - \tfrac{1}{2} \tr(S^{-1}\Omega) + C_1$ and $C_1$ is a constant independent of $\Omega$. Let $h_{ij}'(\cdot)$, $h_{ij}''(\cdot)$ and $h_{ij}'''(\cdot)$ denote the first, second and third derivatives of $h_{ij}(\cdot)$ respectively. From properties of the exponential family, $E(y_{ij}) = h_{ij}'(\eta_{ij})$ and $\var(y_{ij}) = h_{ij}''(\eta_{ij}) \geq 0$. Define 
\begin{equation*} 
g_i(\eta_i)= [h_{i1}'(\eta_{i1}), \dots, h_{in_i}'(\eta_{in_i})]^T, \quad 
H_i(\eta_i)= \diag ([h_{i1}''(\eta_{i1}), \dots, h_{in_i}''(\eta_{in_i})]^T).
\end{equation*}

\subsection{Gaussian approximation of $p(b_i|\theta_G, y_i)$: first approach}
In the first approach, we first find a Gaussian approximation to the likelihood $p(y_i|\eta_i)$ by considering a second-order Taylor expansion about some estimate $\hat{\eta}_i$. This is then combined with $p(b_i|\omega)$ to obtain a Gaussian approximation to $p(b_i|\theta_G, y_i)$. We have $\nabla_{\eta_i} \log p(y_i|\eta_i) = y_i - g_i(\eta_i)$ and $\nabla_{\eta_i}^2 \log p(y_i|\eta_i) = -H_i(\eta_i)$. Thus
\begin{equation*}
\log p(y_i|\eta_i) \approx \log p(y_i|\hat{\eta}_i)  + (\eta_i - \hat{\eta}_i)^T (y_i - g_i(\hat{\eta}_i)) - \tfrac{1}{2}(\eta_i - \hat{\eta}_i)^T H_i(\hat{\eta}_i) (\eta_i - \hat{\eta}_i).
\end{equation*}
As $p(b_i|\theta_G, y_i) \propto \exp\{ \log p(y_i|\eta_i) + \log p(b_i|\omega)  \}$
\begin{multline*}
\propto \exp\{ b_i^T Z_i^T (y_i - g_i(\hat{\eta}_i))  -  \tfrac{1}{2} (X_i \beta + Z_i b_i - \hat{\eta}_i)^T H_i(\hat{\eta}_i) (X_i \beta + Z_i b_i - \hat{\eta}_i) - \tfrac{1}{2}b_i^T \Omega b_i \},
\end{multline*}
$b_i|\theta_G, y_i \sim N(\lambda_i, \Lambda_i)$ approximately, where
\begin{equation} \label{lambdas1}
\Lambda_i = (\Omega + Z_i^T H_i(\hat{\eta}_i) Z_i)^{-1}, \quad 
\lambda_i = \Lambda_i  Z_i^T \{ y_i - g_i(\hat{\eta}_i) + H_i(\hat{\eta}_i)(\heta_i - X_i \beta) \}.
\end{equation}

\subsubsection{Estimate of natural parameter}
Next, we select $\heta_i$ at which the Taylor expansion is evaluated. One possibility is to let $\heta_i$ be the mode of $\log p(y_i|\eta_i)$ or the maximum likelihood estimate (MLE). Then $\nabla_{\eta_i} \log p(y_i|\eta_i)|_{\eta = \heta_i} = y_i - g_i(\heta_i) = 0$ so that the expected value of $y_i$ is equal to its observed value and $\lambda_i$ simplifies to $\Lambda_i Z_i^T H_i(\hat{\eta}_i) (\heta_i - X_i \beta)$. As the $n_i$ observations in $y_i$ are independent, we can evaluate independently each element,  $\heta_{ij}^\ML= \argmax_{\eta_{ij} \in \mathbbm{R}} \;p(y_{ij}|\eta_{ij})$. The estimate $\heta_{ij}^\ML$ is data-based (depends only on $y_{ij}$). However, $\heta_{ij}^\ML$ may not exist for boundary values of $y_{ij}$. For instance, $\heta_{ij}^\ML = \logit(y_{ij})$ for the Bernoulli model, which does not exist when $y_{ij} = 0$ or 1. In such cases, we may consider the ``MLE" at infinity. Lemma \ref{lemA} states the limits of some important quantities in $\lambda_i$ as $\eta_{ij} \rightarrow \pm \infty$, which leads to the definition of $c_i$ as the ``MLE" on the extended real line. From Theorem \ref{thmA}, the limit of $\lambda_i$ always exist as $\heta_i \rightarrow c_i$ and is zero under certain conditions. Proofs of Lemma \ref{lemA} and Theorem \ref{thmA} are given in the supplementary material.

\begin{lemma}\label{lemA}
For $i=1, \dots, n$, $j=1. \dots, n_i$, if $\heta_{ij}^\ML$ does not exist for an observation $y_{ij}$, then $\lim_{\eta_{ij} \rightarrow c} h_{ij}'(\eta_{ij}) = y_{ij}$, $\lim_{\eta_{ij} \rightarrow c} h_{ij}''(\eta_{ij}) = 0$ and $\lim_{\eta_{ij} \rightarrow c}  h_{ij}''(\eta_{ij}) \eta_{ij} = 0$ for either $c=-\infty$ or $c=\infty$.
\end{lemma}

\begin{definition}
For $i=1, \dots, n$, $c_i = [c_{i1}, \dots, c_{in_i}]^T$, where $c_{ij} = \heta_{ij}^\ML$ if $\heta_{ij}^\ML$ exists and $c_{ij} = \lim_{h_{ij}'(\eta_{ij}) \rightarrow y_{ij}} \eta_{ij}$ if $\heta_{ij}^\ML$ does not exist, for $j=1, \dots, n_i$.
\end{definition}

\begin{theorem} \label{thmA}
For $i=1, \dots, n$, the limit of $\lambda_{i}$ exists as $\heta_i \rightarrow c_i$, and is zero if $\heta_{ij}^\ML$ does not exist for each $j=1, \dots, n_i$.
\end{theorem}

\begin{corollary} \label{corA}
For $i=1, \dots, n$. if $\heta_i \rightarrow c_i$, then $\lambda_i \rightarrow 0$
\begin{enumerate}[(i)]
\item for the Poisson GLMM with log link if $y_i = 0$.
\item for the binomial GLMM with logit link if the elements in $y_i$ are either $0$ or $m_{ij}$.
\end{enumerate}
\end{corollary}
\begin{proof}
For the Poisson GLMM with log link, $\heta_{ij}^\ML = \log(y_{ij})$ which is undefined if $y_{ij} = 0$. For the binomial GLMM with logit link, $\heta_{ij}^\ML = \logit(y_{ij}/m_{ij})$ which is undefined if $y_{ij}$ is 0 or $m_{ij}$. From Theorem \ref{thmA}, $\lambda_i \rightarrow 0$ if $\heta_{ij}^\ML$ does not exist for each $j=1, \dots, n_i$.
\end{proof}

From Corollary \ref{corA}, the (location) noncentered parametrization is preferred for Poisson and binomial models when $y_{ij}$ lies on the support boundary, and is always preferred for Bernoulli models. If we set $\heta_i \rightarrow c_i$, then both $\Lambda_i$ and the limit of $\lambda_i$ can be evaluated in principle from Theorem \ref{thmA}. We have tried this approach but found that it results in unstable stochastic variational algorithms if parameter initialization is far from convergence. For instance, if $y_i=0$ for the Poisson GLMM, then $H_i(\heta_i) = \diag(y_i) =0$, which implies that $\lambda_i \rightarrow 0$, $\Lambda_i = \Omega^{-1}$ and $b_i = L_i \tilb_i$. If the estimate of $\omega$ is too small, then estimates for $\Lambda_i = \Omega^{-1} = (WW^T)^{-1}$ and $b_i$ may experience overflow. This issue is alleviated if $H_i(\heta_i)$ is positive definite. While this problem can likely be resolved with better parameter initialization, we prefer non-problem specific initialization. 

To avoid undefined estimates of $\eta_{ij}$ at the support boundary of $y_{ij}$, we ``regularize" the MLE using a Bayesian approach. Adopting a non-informative Jeffreys prior for $\mu_{ij}$, 
\begin{equation*}
p(\mu_{ij}|y_{ij})  \propto p(y_{ij}|\mu_{ij}) p(\mu_{ij})
\propto p(y_{ij}|\mu_{ij}) \sqrt{|\mathcal{I}(\mu_{ij})|},
\end{equation*}
and we use the posterior mean $E(\eta_{ij}|y_{ij}) = E(g(\mu_{ij})|y_{ij})$ to estimate $\eta_{ij}$. \cite{Firth1993} showed that for an exponential family model, the posterior mode corresponding to Jeffreys prior is equivalent to a penalized MLE, where the penalty of $\frac{1}{2} \log |\mathcal{I}(\mu_{ij})|$ helps to remove the asymptotic bias in the MLE. Consider $y_{ij} \sim \text{Poisson}(\mu_{ij})$ with $\log \mu_{ij} = \eta_{ij}$. The Jeffreys prior for $\mu_{ij}$ is $p(\mu_{ij}) \propto \mu_{ij}^{-1/2}$ and $p(\mu_{ij}|y_{ij})$ is Gamma$(y_{ij} +0.5,1)$. Hence 
\begin{equation*}
\heta_{ij}  = E(\eta_{ij}|y_{ij}) = E(\log(\mu_{ij})|y_{ij}) = \psi(y_{ij} +0.5),
\end{equation*}
where $\psi(\cdot)$ is the digamma function. If $y_{ij} \sim \text{binomial}(m_{ij}, p_{ij})$, then the Jeffreys prior for $p_{ij}$ is Beta(0.5, 0.5) and $p(p_{ij}|y_{ij})$ is Beta$(y_{ij} +0.5, m_{ij} - y_{ij} + 0.5)$. Hence
\begin{multline*}
\heta_{ij} = E(\eta_{ij}|y_{ij}) = E(\text{logit}(p_{ij}) |y_{ij}) = E(\log(p_{ij})|y_{ij}) - E(\log(1-p_{ij})|y_{ij}) \\
= \psi(y_{ij} + 0.5) - \psi(m_{ij} - y_{ij} + 0.5).
\end{multline*}
If $y_{ij} \sim \text{Bernoulli}(p_{ij})$, then $m_{ij}=1$ and $\heta_{ij} = \psi(y_{ij} + 0.5) - \psi(1.5 - y_{ij})$. Figure \ref{fig_reg_est} shows that the regularized estimates of $\heta_{ij}$ are indistinguishable from the MLEs for values of $y_{ij}$ that do not lie on the support boundary. However, by using these regularized estimates, $\lambda_i$ will be close to but is not zero even if $y_i$ lies on the support boundary for the Poisson and binomial GLMMs (see Table \ref{tab_reg_est}). 
\begin{figure}
\centering
\includegraphics[height=220pt, angle=270]{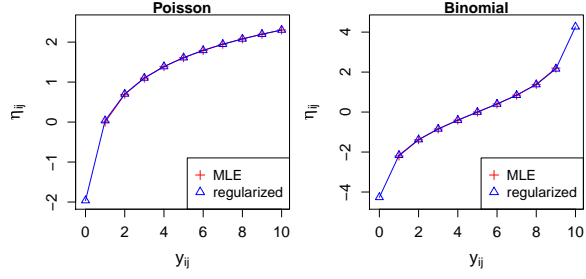}
\caption{Plot of maximum likelihood and regularized estimates of $\heta_{ij}$ against $y_{ij}$ for Poisson and binomial ($m_{ij} =10$) models.} \label{fig_reg_est}
\end{figure}
\begin{table}
\caption{\label{tab_reg_est} Values of some quantities in Lemma \ref{lemA} evaluated using regularized estimates of $\heta_{ij}$ compared to the ``MLEs" at infinity (in brackets).} 
\centering
\fbox{
\begin{tabular}{lccccc}
& Poisson & \multicolumn{2}{c}{Binomial(10, $p_{ij}$)}  & \multicolumn{2}{c}{Bernoulli($p_{ij}$)}\\
& $y_{ij}=0$ & $y_{ij}=0$ & $y_{ij}=m_{ij}$ & $y_{ij}=0$ & $y_{ij}=1$ \\
\hline
$\heta_{ij}$ & -1.96 ($-\infty$) & -4.27 ($-\infty$) & 4.27 ($\infty$) & -2  ($-\infty$) & 2  ($\infty$)\\
$h_{ij}'(\heta_{ij})$ & 0.14 (0) & 0.14 (0) & 9.86 (10) & 0.12 (0) & 0.88 (1) \\
$h_{ij}''(\heta_{ij})$ & 0.14 (0) & 0.14 (0) & 0.14 (0) & 0.10 (0) & 0.10 (0) \\
$h_{ij}''(\heta_{ij}) \heta_{ij}$ & -0.28 (0)  & -0.58 (0) & 0.58 (0) & -0.21 (0) & 0.21 (0)\\
\end{tabular}}
\end{table}

\subsubsection{Difficulties with the first approach} \label{Sec_diff}
From our experiments, a Gaussian approximation of $p(b_i|y_i, \theta_G)$ based on Taylor expansion of $p(y_i|\eta_i)$ about regularized estimate $\heta_i$ generally works better for Poisson and binomial models than Bernoulli models. To understand this phenomenon, consider the Fisher information $\mathcal{I}(\eta_{ij})$, which measures the amount of information $y_{ij}$ carries about $\eta_{ij}$. For the one-parameter exponential family $p(y_{ij}|\eta_{ij})$, $\mathcal{I}(\eta_{ij}) = h''(\eta_{ij}) = \var(y_{ij})$. We have $\mathcal{I}(\eta_{ij}) = \mu_{ij}$ for Poisson($\mu_{ij}$) and $\mathcal{I}(\eta_{ij}) = m_{ij} p_{ij} (1-p_{ij})$ for Binomial($m_{ij}, p_{ij}$). Therefore, $y_{ij}$ is more informative about $\eta_{ij}$ when $\mu_{ij}$ is large for Poisson models, and when $m_{ij}$ is large and $p_{ij}$ is far from 0 and 1 for binomial models. Hence, we expect the first approach to perform well when $y_{ij}$ is large for Poisson models, and when $m_{ij}$ is large and $y_{ij}/m_{ij}$ is far from 0 and 1 for binomial models. For the Bernoulli model, $m_{ij}$ is at its lowest value of 1. Hence $y_{ij}$ carries less information about $\eta_{ij}$ for Bernoulli models and it is difficult to obtain an accurate approximation of $p(b_i|y_i, \theta_G)$ using the first approach. The second approach below is computationally more intensive but it works better for Bernoulli models, binomial models with small $m_{ij}$ or $y_{ij}/m_{ij}$ close to 0 or 1, and Poisson models with many zero counts.

\subsection{Gaussian approximation of $p(b_i|\theta_G, y_i)$: second approach}
Consider a Taylor expansion of $\log p(b_i|\theta_G, y_i)$  about some estimate $\hat{b}_i$ of $b_i$. We have
\begin{equation*}
\begin{gathered}
\nabla_{b_i} \log p(b_i|\theta_G, y_i) = Z_i^T (y_i - g_i(\eta_i))-  \Omega b_i,  \quad 
\nabla_{b_i}^2 \log p(b_i|\theta_G, y_i) = - Z_i^T H_i(\eta_i) Z_i -  \Omega,
\end{gathered}
\end{equation*}
and
\begin{equation*}
\begin{aligned}
p(b_i|\theta_G, y_i) 
&\propto \exp\big[   b_i^T \{Z_i^T (y_i -  g_i(X_i \beta + Z_i \hat{b}_i)) - \Omega \hat{b}_i\} \\
& \quad - \tfrac{1}{2}(b_i - \hat{b}_i) ^T \{Z_i^T H_i (X_i \beta + Z_i \hat{b}_i) Z_i +  \Omega\} (b_i - \hat{b}_i)\big].
\end{aligned}
\end{equation*}
This implies a Gaussian approximation with $\Lambda_i^{-1} = Z_i^T H_i(X_i \beta+ Z_i \hat{b}_i) Z_i +  \Omega$ and
$\lambda_i = \hat{b}_i + \Lambda_i [Z_i^T \{y_i -  g_i(X_i \beta+ Z_i \hat{b}_i)\} - \Omega \hat{b}_i]$. If we choose $\hat{b}_i$ to be the mode of $p(b_i|\theta_G, y_i)$, then $\hat{b}_i$ satisfies $ Z_i^T (y_i - g_i(X_i \beta + Z_i \hat{b}_i)) = \Omega \hat{b}_i$ and
\begin{equation} \label{lambdas2}
\begin{gathered}
\lambda_i = \hat{b}_i, \quad \Lambda_i = \{Z_i^T H_i(X_i \beta+ Z_i \hat{b}_i) Z_i +  \Omega\}^{-1}.
\end{gathered}
\end{equation}
This approach yields more accurate estimates of $\lambda_i$ and $\Lambda_i$, which translates into better variational approximations, but additional computation is required to find the mode $\hat{b}_i$. Using the Newton-Raphson method, we start with initial estimate $b_i^{(0)}$ and iterate
\begin{equation*}
\begin{aligned}
b_i^{(t+1)} = b_i^{(t)} +  (Z_i^T H_i(X_i \beta + Z_i b_i^{(t)}) Z_i +  \Omega)^{-1} \{ Z_i^T (y_i - g_i(X_i \beta + Z_i b_i^{(t)}))-  \Omega b_i^{(t)} \}
\end{aligned}
\end{equation*}
until convergence. If $n_i < r$, we set $b_i^{(0)}= 0$. Otherwise, we obtain $b_i^{(0)}$ by using the regularized estimate $\hat{\eta}_i$ described in the first approach. As $\hat{\eta}_i \approx X_i \beta + Z_i b_i$, an estimate of $b_i$ is $b_i^{(0)} = (Z_i^TZ_i)^{-1}Z_i^T (\hat{\eta}_i - X_i \beta)$. We find this initialization to be highly useful in avoiding convergence issues. Monitoring $\log p(b_i|\theta_G, y_i)$ at each iteration, updates are terminated if the rate of increase of $\log p(b_i|\theta_G, y_i)$ is smaller than $10^{-4}$.

\section{Reparametrized variational Bayes inference} \label{sec_variational}
We consider a variational approximation for the reparametrized model of the form in \eqref{pdt2}, where $q(\theta_G)$ and $q(\tilb_i)$ are Gaussian. Hence, $q(\tiltheta)$ is $N(\mu, \Sigma)$, where $\Sigma$ is a block diagonal matrix; the first $n$ blocks are of order $r$ each and the last block is of order $g$. Let $C C^T$ be the unique Cholesky factorization of $\Sigma$ where $C$ is lower triangular with positive diagonal elements. Then $q(\tiltheta) = q(\tiltheta | \mu, C)$ and $\{\mu, C \}$ are optimized so that $D_{\KL}(q||p) = \int q(\tiltheta | \mu, C) \log \{ q(\tiltheta | \mu, C) / p(\tiltheta|y) \} \d\tiltheta$ is minimized. Since $D_{\KL}(q||p) \geq 0$, 
\begin{equation} \label{lowerbound}
\begin{aligned}
\log p(y) \geq E_q \{\log p(y,\tiltheta) - \log q(\tiltheta | \mu, C) \} =  \L (\mu, C) = \L,
\end{aligned}
\end{equation}
where $E_q$ denotes expectation with respect to $q(\tiltheta | \mu, C)$. Minimizing the KL divergence is equivalent to maximizing a lower bound $\L$ on $\log p(y)$ with respect to $\{\mu, C\}$.

For non-conjugate models such as GLMMs, $E_q\{ \log p(y, \tiltheta) \}$ cannot be evaluated in closed form and $\L$ is intractable. To overcome this limitation, we maximize $\L$ with respect to $\{\mu, C \}$  using stochastic variational inference \citep{Titsias2014}. This approach is based on stochastic approximation \citep{Robbins1951}, where at each iteration $t$, the variational parameters are updated by taking a small step $\rho_t$ in the direction of the stochastic gradients,
\begin{equation} \label{updates}
\begin{aligned}
\mu^{(t)} = \mu^{(t-1)} + \rho_t \hnabla_{\mu} \L,  \quad 
\vech(C^{(t)}) = \vech(C^{(t-1)}) + \rho_t \hnabla_{\vech(C)} \L.
\end{aligned}
\end{equation}
The stochastic gradients $\hnabla_{\mu} \L$ and $\hnabla_{\vech(C)} \L$ are unbiased estimates of the true gradients $\nabla_{\mu} \L$ and $\nabla_{\vech(C)} \L$ respectively. Under mild regularity conditions, the stochastic approximation algorithm will converge to a local maximum if the stepsize $\{\rho_t\}$ satisfy $\sum_t \rho_t = \infty$, $\sum_t \rho_t^2 < \infty$ and the lower bound is concave \citep{Spall2003}. 

From \eqref{lowerbound}, unbiased estimates of the true gradients can be constructed by sampling $\tiltheta$ from $q(\tiltheta|\mu, C)$. However, this approach often results in estimators with high variance \citep{Paisley2012}. Hence we employ the ``reparametrization trick" \citep{Kingma2014, Rezende2014, Titsias2014}, which introduces an invertible transformation $s = C^{-1} (\tiltheta  - \mu)$. The density of $s$, denoted by $\phi(s)$, is $N(0, I_d)$. Let $\ell(\tiltheta) = \log p(y, \tiltheta)$. From \eqref{lowerbound}, 
\begin{equation} \label{lowerbound2}
\begin{aligned}
\L(\mu, C) & = E_{\phi} \{ \ell(\tiltheta) \} - E_{\phi} \{\log q(\tiltheta | \mu, C)\}.
\end{aligned}
\end{equation}
where $E_\phi$ denotes expectation with respect to $\phi(s)$ and $\tiltheta = C s + \mu$. Unbiased estimates of the true gradients can thus be constructed by sampling $s$ from $\phi(s)$ instead of $\tiltheta$ from $q(\tiltheta| \mu, C)$. The advantage of this reparametrization is that $\ell(\tiltheta)$ is now a function of $\{\mu, C\}$, and this enables gradient information from $\ell(\tiltheta)$ to be utilized effectively.

\subsection{Unbiased estimates of stochastic gradients}
We show that several unbiased estimators of $\nabla_{\mu} \L$ and $\nabla_{\vech(C)} \L$ can be constructed from $\eqref{lowerbound2}$ but some has nicer properties at convergence. The estimators below are based on a single sample $s$ generated from $\phi(s)$, as one sample often provides sufficient gradient information. It is straightforward to extend the estimators to a larger number of samples by averaging. Derivation details are given in the supplementary material.

If we evaluate the second term in \eqref{lowerbound2} analytically, then $\L = E_{\phi} \{\ell(\tiltheta)\} + \log|C| + c'$, where $c'$ is a constant that does not depend on $\{\mu, C\}$ and
\begin{equation*}
\nabla_{\mu} \L= E_{\phi} \{\nabla_{\tiltheta} \ell(\tiltheta)\}, \quad 
\nabla_{\vech(C)} \L = E_{\phi} [\vech\{ \nabla_{\tiltheta} \ell(\tiltheta) s^T + C^{-T} \}].
\end{equation*}
 This leads to unbiased estimators, 
\begin{equation} \label{est1}
\hnabla_{\mu} \L_1= \nabla_{\tiltheta} \ell(\tiltheta), \quad 
\hnabla_{\vech(C)} \L_1 = \vech\{ \nabla_{\tiltheta} \ell(\tiltheta) s^T + C^{-T} \}.
\end{equation}
Alternatively, we do not evaluate $E_{\phi} \{\log q(\tiltheta | \mu, C)\}$ analytically but approximate both terms in \eqref{lowerbound2} using the same samples. As $\log q(\tiltheta|\mu, C)$ depends on $\{ \mu, C\}$ directly as well as through $\tiltheta$, we apply chain rule to obtain
\begin{align}
\nabla_{\mu} \L &= E_{\phi} \{\nabla_{\tiltheta} \ell(\tiltheta) - \nabla_{\tiltheta} \log q(\tiltheta|\mu, C) - \nabla_{\mu} \log q(\tiltheta|\mu, C)\},  \label{g1}\\
\nabla_{\vech(C)} \L &= E_{\phi} [\vech \{ \nabla_{\tiltheta} \ell(\tiltheta) s^T - \nabla_{\tiltheta} \log q(\tiltheta|\mu, C) s^T\} - \nabla_{\vech(C)} \log q(\tiltheta|\mu, C) ], \label{g2}
\end{align}
where we have $\nabla_{\tiltheta} \log q(\tiltheta|\mu, C) = - \nabla_{\mu} \log q(\tiltheta|\mu, C)\} = - C^{-T} s$ and $\nabla_{\vech(C)} \log q(\tiltheta|\mu, C) = \vech\{C^{-T} (ss^T - I_d)\}$. If we use all the terms in \eqref{g1} and \eqref{g2}, the estimators in \eqref{est1} will be recovered after simplification. However, as $E_{\phi}(s) = 0$ and $E_{\phi} (ss^T) = I_d$, 
\begin{equation*}
\begin{gathered}
E_{\phi} \{ \nabla_{\tiltheta} \log q(\tiltheta|\mu, C) \} = 
E_{\phi} \{ \nabla_{\mu} \log q(\tiltheta|\mu, C) \} = 0, \quad E_{\phi} \{\nabla_{\vech(C)} \log q(\tiltheta|\mu, C)\} = 0.
\end{gathered}
\end{equation*}
Alternatively, note that $E_{\phi} \{ \nabla_{\mu} \log q(\tiltheta|\mu, C) \} = 0$ and $E_{\phi} \{\nabla_{\vech(C)} \log q(\tiltheta|\mu, C)\} = 0$ as they form the expectation of the score. Unbiased estimators can thus be constructed by omitting the last term in \eqref{g1} and \eqref{g2}. We thus obtain
\begin{equation} \label{est2}
\begin{aligned}
\hnabla_{\mu} \L_2 = \nabla_{\tiltheta} \ell(\tiltheta) + C^{-T} s, \quad
\hnabla_{\vech(C)} \L_2 = \vech \{\nabla_{\tiltheta} \ell(\tiltheta) s^T+ C^{-T}s s^T\}.
\end{aligned}
\end{equation}
A third unbiased estimator, $\hnabla_{\mu} \L_3 = \nabla_{\tiltheta} \ell(\tiltheta) - C^{-T} s$, can be obtained by omitting the second term in \eqref{g1}. Next, we show that $\hnabla_{\mu} \L_2$ and $\hnabla_{\vech(C)} \L_2 $ have smaller variance at convergence.

Consider a second-order Taylor approximation to $\ell(\tiltheta)$ at the posterior mode $\tiltheta^*$,
\begin{equation*}
\ell(\tiltheta) \approx \ell(\tiltheta^*) + (\tiltheta - \tiltheta^*)^T \nabla^2_{\tiltheta} \ell(\tiltheta^*) (\tiltheta - \tiltheta^*)/2.
\end{equation*}
This implies that $p(\tiltheta|y)$ can be approximated by $N(\tiltheta^*, -\{\nabla^2_{\tiltheta} \ell(\tiltheta^*)\}^{-1})$. Differentiating the Taylor approximation with respect to $\tiltheta$ yields
\begin{equation} \label{grad}
\nabla_{\tiltheta} \ell(\tiltheta) \approx \nabla^2_{\tiltheta} \ell(\tiltheta^*) (\tiltheta - \tiltheta^*).
\end{equation}
Since $q(\tiltheta|\mu, \Sigma)$ is a Gaussian approximation to $p(\tiltheta|y)$, $\mu \approx \tiltheta^*$ and $\Sigma \approx -\{\nabla^2_{\tiltheta} \ell(\tiltheta^*)\}^{-1}$ at convergence. Thus, $\nabla_{\tiltheta} \ell(\tiltheta) \approx  - \Sigma^{-1}(\tiltheta - \mu) = - C^{-T} s$ and
\begin{equation*}
\begin{aligned}
\hnabla_{\mu} \L_1 &\approx - C^{-T} s,\\
\hnabla_{\mu} \L_2 &\approx 0, 
\end{aligned} 
\qquad
\begin{aligned}
\hnabla_{\vech(C)} \L_1 & \approx \vech\{ -C^{-T} s s^T + C^{-T} \}, \\
\hnabla_{\vech(C)} \L_2 & \approx 0.
\end{aligned}
\end{equation*}
The estimators $\hnabla_{\mu} \L_2$ and $\hnabla_{\vech(C)} \L_2 $ are close to zero as the contributions from $\ell(\tiltheta)$ and $\log q(\tiltheta|\mu, C)$ cancel out each other when the stochastic approximation algorithm is close to convergence. However $\hnabla_{\mu} \L_1$ and $\hnabla_{\vech(C)} \L_1 $ still contain a certain amount of noise and $\hnabla_{\mu} \L_3 \approx - 2C^{-T} s$ is even noisier than $\hnabla_{\mu} \L_1$. From \eqref{grad}, 
\begin{equation*}
\begin{aligned}
\cov_{\phi}(\hnabla_{\mu} \L_1) & = \cov_q (\nabla_{\tiltheta} \ell(\tiltheta)) \approx \nabla^2_{\tiltheta} \ell(\tiltheta^*) \Sigma \{\nabla^2_{\tiltheta} \ell(\tiltheta^*)\}^T \approx \Sigma^{-1}, \\
\cov_{\phi}(\hnabla_{\mu} \L_2) & = \cov_q (\nabla_{\tiltheta} \ell(\tiltheta) + \Sigma^{-1} \tiltheta) \approx \{\nabla^2_{\tiltheta} \ell(\tiltheta^*) + \Sigma^{-1}\} \Sigma \{\nabla^2_{\tiltheta} \ell(\tiltheta^*) + \Sigma^{-1}\}^T \approx 0,
\end{aligned}
\end{equation*}
which supports the claim that $\hnabla_{\mu} \L_2$ is less noisy than $\hnabla_{\mu} \L_1$ at convergence. From our experience, the less noisy estimators contributes greatly to improved convergence of the stochastic variational algorithm \citep[see also][]{Roeder2017}.

\subsection{Algorithm implementation} \label{sec:Alg_imp}
As the update for $\vech(C)$ in \eqref{updates} does not ensure diagonal elements of $C$ remain positive, we introduce lower triangular matrix $C^*$ such that $C^*_{ii} = \log(C_{ii})$ and $ C^*_{ij} = C_{ij}$ if $i\neq j$, and apply stochastic gradient updates to $C^*$ instead. Let $D^C$ be a diagonal matrix of order $d(d+1)/2$ with diagonal given by $\vech(J^C)$ and $J^C$ be a $d \times d$ matrix with diagonal given by $\diag(C)$ and ones everywhere else. Then $\nabla_{\vech(C^*)} \L = D_C \nabla_{\vech(C)} \L$. The stochastic variational algorithm using $\hnabla_{\mu} \L_2$ and $\hnabla_{\vech(C)} \L_2$ is summarized in Algorithm 1.

\begin{Algorithm}
\centering
\parbox{0.9\textwidth}{
\hrule
Initialize $\mu^{(1)} = 0$ and $C^{(1)} = \text{blockdiag}(I_{nr}, 0.1I_g)$. For $t=1, \dots, T$,
\begin{enumerate}[1.]
\vspace{-2mm}
\itemsep 0em 
\item Generate $s \sim N(0, I_d)$.
\item Compute $\tiltheta^{(t)} = C^{(t)} s + \mu^{(t)}$ and $\mathcal{G}^{(t)} = \nabla_{\tiltheta} \ell(\tiltheta^{(t)}) + {C^{(t)}}^{-T} s$.
\item Update $\mu^{(t+1)}=\mu^{(t)} + \rho_t \mathcal{G}^{(t)}$ and $\vech(C'^{(t+1)}) = \vech(C'^{(t)}) + \rho_t D_C\vech  \{\mathcal{G}^{(t)} s^T\}.$
\item Obtain $C^{(t+1)}$ from ${C^*}^{(t+1)}$.
\end{enumerate}
\hrule}
\caption{RVB algorithm.}\label{Alg1}
\end{Algorithm} 

For computing the stepsize $\rho_t$, we use Adam \citep{Kingma2014}, which automatically adapts to different parameters. There are some parameters in Adam which can be used to fine-tune the stepsize but we find that the default values work satisfactorily. For diagnosing the convergence of Algorithm \ref{Alg1}, we compute an unbiased estimate of $\L$ at each iteration, $\hat{\L} = \ell(\tiltheta) - \log q(\tiltheta|\mu, C)$ where $\tiltheta$ is computed in step 2 of Algorithm \ref{Alg1}, and $\mu$ and $C$ are the current estimates. As these estimates are stochastic, we consider the lower bound averaged over every 1000 iterations and monitor the path of these means. At the beginning, the means tend to increase monotonically. However, when the algorithm is close to convergence, these means start bouncing around the true maximum. Hence we consider the gradient of a least square regression line fitted to the past $\tau$ means and Algorithm 1 is terminated once the gradient becomes negative. This process is illustrated in Figure \ref{gradlr}. For the experiments, we set $\tau=5$. When the number of means is less than $\tau$, we only fit the regression lines to existing means.
\begin{figure}
\centering
\makebox{\includegraphics[height=175pt, angle=270]{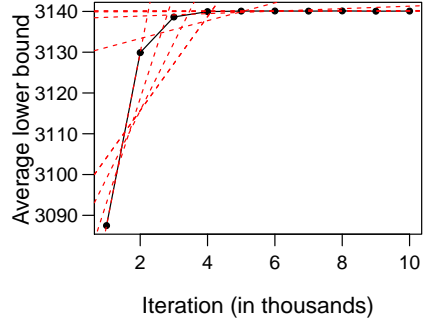}}
\caption{Plot of average lower bound against iteration for epilepsy data Model II. Dotted lines are the fitted least square regression lines, whose gradients decrease to zero.} \label{gradlr}
\end{figure}

\subsection{Marginal posterior distribution of the local variables} \label{sec:b_i}
The RVB algorithm does not return the posterior distributions of $\{b_i\}$ directly as it considers a Gaussian approximation for transformed local variables $\{\tilb_i\}$. Marginal posterior distributions of $\{b_i\}$ can be estimated using simulation: 
\begin{enumerate}
\itemsep 0em 
\item Generate $\theta_G$ from $q(\theta_G)$. 
\item For $i=1, \dots, n$,
\vspace{-2mm}
\begin{enumerate}
\item Generate $\tilb_i$ from $q(\tilb_i)$.
\item Compute $\lambda_i$ and $L_i$ from $\theta_G$ using \eqref{lambdas1} or \eqref{lambdas2}. 
\item Compute $b_i = L_i \tilb_i + \lambda_i$.
\end{enumerate}
\end{enumerate}
While this takes more work, a possible advantage is that the posterior distributions of $\{b_i\}$ are not constrained to be Gaussian and hence may be better able to accommodate any skewness present in the true marginal posterior $p(b_i|y)$.

\section{Gradient of the log joint density} \label{sec_gradient}
To implement Algorithm 1, we require $\nabla_{\tiltheta}\ell(\tiltheta) = [\nabla_{\tilb_1}\ell(\tiltheta), \dots, \nabla_{\tilb_n}\ell(\tiltheta), \nabla_{\theta_G}\ell(\tiltheta)]$. As $b_i = L_i \tilb_i + \lambda_i$, $p(\tilb_i|\theta_G) = p(b_i|\theta_G)|L_i|$. Hence $p(y, \tiltheta) = p(y, \theta) \prod_{i=1}^n |L_i|$ and the log joint density of the model in \eqref{mod} after reparametrization is
\begin{equation*}
\ell(\tiltheta) = \log p(\theta_G) + \sum_{i=1}^n \{\log p(y_i, b_i| \theta_G) + \log |L_i|\}. 
\end{equation*}
The gradients $\nabla_{\tilb_i}\ell(\tiltheta)$ for $i=1, \dots, n$, and $\nabla_{\theta_G}\ell(\tiltheta)$ are derived in Theorem \ref{thm2}. This result is applicable to any model of the form in \eqref{mod}. Proof of Theorem \ref{thm2} is given in the supplementary material.
\begin{theorem} \label{thm2}
For $i=1, \dots, n$, let $a_i = \nabla_{b_i} \log p(y_i, b_i|\theta_G)$, $B_i = L_i^T a_i \tilb_i^T $ and $\tilde{B}_i = \bar{B}_i + \bar{B}_i^T - \dg(B_i)$. Then
$\nabla_{\tilb_i}\ell(\tiltheta) = L_i^T a_i$ for $i=1, \dots, n$. Suppose $\theta_G$ is partitioned as $[\theta_{G_1}^T, \dots, \theta_{G_M}^T]^T$. For $m=1, \dots, M$,
\begin{equation*}
\begin{aligned}
\nabla_{\theta_{G_m}} \ell(\tiltheta) &= \sum_{i=1}^n (\nabla_{\theta_{G_m}} \lambda_i) a_i + \frac{1}{2}\sum_{i=1}^n \{\nabla_{\theta_{G_m}}\vec(\Lambda_i)\} \vec (\Lambda_i^{-1} + L_i^{-T} \tilde{B}_i L_i^{-1}) \\
&\quad + \sum_{i=1}^n \nabla_{\theta_{G_m}} \log p(y_i, b_i|\theta_G) + \nabla_{\theta_{G_m}} \log p(\theta_G).
\end{aligned}
\end{equation*}
\end{theorem}

\subsection{Gradients of reparametrized generalized linear mixed models} 
For the GLMM in Section \ref{sec: app_GLMMs}, after deriving $\lambda_i$ and $\Lambda_i$ as in \eqref{lambdas1} or \eqref{lambdas2}, we reparametrize the model by transforming each $b_i$ so that $\tilb_i = L_i^{-1} (b_i - \lambda_i)$, where $L_iL_i^T = \Lambda_i$. From Theorem \ref{thm2}, to evaluate $\nabla_{\tilb_i} \ell(\tiltheta)$, we require $a_i = Z_i^T (y_i - g_i(\eta_i) )- \Omega b_i$ for $i=1, \dots, n$. 
As $\theta_G$ is partitioned into two components, $\beta$ and $\omega$, we use Theorem \ref{thm2} to find $\nabla_{\beta} \ell(\tiltheta)$ and $\nabla_{\omega} \ell(\tiltheta)$. We have $\nabla_{\beta} \log p(\theta_G) =  - \beta/ \sigma_\beta^2$, $\nabla_{\beta} \log p(y_i, b_i|\theta_G) = X_i^T (y_i - g_i(\eta_i) )$ and the results in Lemma \ref{lem2}.
\begin{lemma} \label{lem2}
Let $D^W$ be a diagonal matrix of order $r(r+1)/2$ where the diagonal is given by $\vech(J^W)$, and $J^W$ is an $r \times r$ matrix with $J^W_{ii} = W_{ii}$ and $J^W_{ij} = 1$ if $i \neq j$. In addition, let $u$ be a vector of length $r$ where the $i$th element is $r-i+2$. Then 
\begin{enumerate}
\itemsep 0em
\item $\nabla_{\omega} \log p(y_i, b_i|\theta_G) = D^W\vech(W^{-T} - b_ib_i^T W) $,
\item $\nabla_\omega  \log p(\theta_G) = D^W  \vech \{(\nu-r-1) W^{-T} - S^{-1} W\} + \vech(\diag(u))$.
\end{enumerate}
\end{lemma}

Next, we need to find $\nabla_\beta \vec(\Lambda_i)$, $\nabla_{\beta} \lambda_i$, $\nabla_\omega \vec(\Lambda_i)$ and $\nabla_\omega \lambda_i$ for approaches 1 and 2. These results are then substituted into Theorem \ref{thm2} to obtain the results in Lemma \ref{lem3} and \ref{lem4}. Proofs of Lemmas \ref{lem2},  \ref{lem3} and \ref{lem4} are given in the supplementary material. 

\begin{lemma} \label{lem3}
For approach 1, $\nabla_\beta \ell(\tiltheta) = \sum_{i=1}^n X_i^T (y_i - g_i(\eta_i) - H_i(\hat{\eta}_i) Z_i \Lambda_i a_i ) - \beta/ \sigma_\beta^2$ and
\begin{equation*}
\begin{aligned}
\nabla_\omega \ell(\tiltheta) &=  \vech(\diag(u)) + D^W  \vech \big[ (n + \nu-r-1) W^{-T} - S^{-1} W  \\
& \quad - \sum_{i=1}^n (b_ib_i^T + \Lambda_i a_i \lambda_i^T + \lambda_i a_i^T\Lambda_i  + \Lambda_i + L_i\tilde{B}_i L_i^T) W \big].
\end{aligned}
\end{equation*}
\end{lemma}

\begin{lemma} \label{lem4}
For approach 2, let $\alpha_i = \tfrac{1}{2}h_i'''(X_i\beta + Z_i \hat{b}_i) \odot \diag\{ Z_i ( \Lambda_i + L_i \tilde{B}_i L_i^T ) Z_i^T \}$.
\begin{equation*}
\begin{aligned}
\nabla_\beta \ell(\tiltheta) &= \sum_{i=1}^n X_i^T \{y_i - g_i(\eta_i) - H_i(X_i \beta + Z_i \hat{b}_i) Z_i \Lambda_i (a_i - Z_i^T \alpha_i) - \alpha_i \} - \beta/ \sigma_\beta^2. \\
\nabla_\omega \ell(\tiltheta) &= \vech(\diag(u)) + D^W \vech[ (n + \nu-r-1) W^{-T} - S^{-1} W \\
&- \sum_{i=1}^n \{b_ib_i^T + \Lambda_i (a_i - Z_i^T \alpha_i ) \hat{b}_i^T + \hat{b}_i (a_i - Z_i^T \alpha_i )^T \Lambda_i + \Lambda_i + L_i \tilde{B}_i L_i^T \} W ].
\end{aligned}
\end{equation*}
\end{lemma}

\section{Extension to large data sets} \label{sec:ext}
In this section, we discuss how the RVB algorithm can be extended to large data using a divide and recombine strategy \citep{Broderick2013, Tran2016}. Suppose the $n$ observations in the data are partitioned into $V$ parts such that $y = (y^1, \dots, y^V)^T$ and let $\tilb^v$ be the set of transformed local variables corresponding to $y^v$. The true posterior, 
\begin{multline*}
p(\tiltheta|y) = p(\tilb,\theta_G|y) 
\propto p(\theta_G)\prod_{v=1}^V  p(y^v|\tilb^v, \theta_G) p(\tilb^v|\theta_G)\\
\propto \frac{\prod_{v=1}^V \{p(y^v|\tilb^v, \theta_G) p(\tilb^v|\theta_G) p(\theta_G) \}}{p(\theta_G)^{V-1}} 
\propto \frac{\prod_{v=1}^V p(\theta_G, \tilb^v| y^v)}{p(\theta_G)^{V-1}}.
\end{multline*}
If we estimate $p(\theta_G, \tilb^v| y^v)$ by our variational approximation $q^v (\theta_G) q^v(\tilb^v)$, which is obtained using the portion $y^v$ of the data only, then
\begin{equation*}
p(\tilb,\theta_G|y)  \propto \frac{\prod_{v=1}^V \{q^v (\theta_G) q^v(\tilb^v)\}}{p(\theta_G)^{V-1}} \;\; \text{approximately}.
\end{equation*}
Due to the assumption of independent variational posteriors for the local and global variables, we can integrate out the random effects $\tilb$ on both sides to obtain $p(\theta_G|y) \propto \prod_{v=1}^V q^v (\theta_G) / p(\theta_G)^{V-1}$. Suppose $p(\theta_G)$ is $N(\mu_0, \Sigma_0)$. As $q^v(\theta_G)$ is Gaussian, say $N(\mu_v, \Sigma_v)$,
\begin{equation*}
p(\theta_G|y) \propto  \exp\left[ -\frac{1}{2} \left\{  \sum_{v=1}^V (\theta_G - \mu_v)^T \Sigma_v^{-1} (\theta_G - \mu_v) - (V-1) (\theta_G - \mu_0)^T \Sigma_0^{-1} (\theta_G - \mu_0)  \right\}  \right].
\end{equation*}
Thus the distribution of $p(\theta_G|y)$ can be approximated by $N(\mu, \Sigma)$, where
\begin{equation*}
\Sigma = \left\{ \sum_{v=1}^V \Sigma_v^{-1} - (V-1)\Sigma_0^{-1} \right\}^{-1}, \quad
\mu = \Sigma \left\{  \sum_{v=1}^V \Sigma_v^{-1} \mu_v - (V-1)  \Sigma_0^{-1} \mu_0 \right\}.
\end{equation*}

Another possible way of extending RVB to large data sets is to fix $q(\tilb_i)$ as $N(0, I_r)$ and optimize only $q(\theta_G)$, since $p(b_i|\theta_G, y)$ is approximately $N(0, I_r)$ after transformation. To increase efficiency, one can also replace $\nabla_{\theta_{G_m}} \ell(\tiltheta)$ by an unbiased estimate, computed by choosing a random sample of $M$ subjects from $\{1, \dots, n\}$ at each iteration, say $\mathcal{S} = \{i_1, \dots, i_M\}$, and replacing $\sum_{i=1}^n$ in Theorem \ref{thm2} by $n\sum_{i \in \mathcal{S}}/M$. Further work is required to determine the accuracy of such an approach.

\section{Experimental results} \label{sec_results}
We present the results of fitting the RVB algorithm to GLMMs and also use a large dataset to investigate the quality of results obtained via dataset partitioning and  parallel processing. Results of the RVB algorithm are compared with GVA \citep{Tan2018} and INLA, with the posterior distributions estimated using MCMC via RStan regarded as ground truth. We refer to the RVB algorithm implemented using approaches 1 and 2 as RVB1 and RVB2 respectively.

GVA similarly approximates $p(\theta|y)$ by $N(\mu, \Sigma)$ and is also based on the stochastic variational algorithm \citep{Titsias2014}. The posterior dependence structure is captured using a Cholesky decomposition $TT^T$ of the precision matrix $\Sigma^{-1}$, where $T$ is lower triangular. $T$ and $\Sigma^{-1}$ are assumed to be sparse matrices of the form,
\begin{equation*} 
\begin{aligned}
T &= \begin{bmatrix}
T_{11}  & \ldots & 0 & 0 \\
\vdots  & \ddots & \vdots & \vdots\\
0 & \ldots & T_{nn} & 0 \\
T_{n+1,1} & \ldots & T_{n+1,n} & T_{n+1,n+1} \\
\end{bmatrix}, \,
\Sigma^{-1} &= \begin{bmatrix}
\Sigma^{-1}_{11} & \ldots & 0 & \Sigma^{-1}_{1,n+1} \\
\vdots & \ddots & \vdots & \vdots \\
0 & \ldots & \Sigma^{-1}_{nn} & \Sigma^{-1}_{n,n+1} \\
\Sigma^{-1}_{n+1,1} & \ldots & \Sigma^{-1}_{n+1,n} & \Sigma^{-1}_{n+1,n+1}.
\end{bmatrix},
\end{aligned}
\end{equation*}
where the block matrices $\Sigma^{-1}_{11}$, $\dots$, $\Sigma^{-1}_{nn}$, $\Sigma^{-1}_{n+1,n+1}$ correspond to $b_1$, $\dots,$ $b_n$, $\theta_G$ respectively. As $\{b_i\}$ are independent conditional on $\theta_G$ a posteriori, a block diagonal structure is assumed for these components while $\Sigma^{-1}_{n+1, i}$ captures the conditional posterior dependence between $b_i$ and $\theta_G$. On the other hand, RVB minimizes posterior dependence between local and global variables through model reparametrization and considers a Cholesky decomposition of $\Sigma$, which is of a block diagonal structure. The number of variational parameters in RVB is smaller than GVA by $nrg$ (corresponding to absence of the off-diagonal blocks $T_{n+1,1}, \dots, T_{n+1,n}$). This reduction can be significant when $n$, $r$ and $g$ are large. GVA and RVB both account for posterior dependence, albeit in different manners, and we observe that RVB can often achieve a better posterior approximation and higher lower bound than GVA at a faster convergence rate. The code for RVB and GVA is written in Julia version 0.6.4 and are available as supplementary material. \cite{Ormerod2012} consider non-Bayesian Gaussian variational approximate inference for GLMMs (also known as GVA), which derives estimates of model parameters by optimizing a lower bound on the log-likelihood. They use Gaussian approximations of the conditional distributions of random effects and then integrate out the random effects using adaptive Gauss-Hermite quadrature.  

For comparability, the same priors are used for all algorithms and we set $\sigma_\beta^2 = 100$. MCMC is performed using either the centered or noncentered parametrization depending on which is more efficient, and we run four chains in parallel with 25,000 iterations each. The first half of each chain is used as warmup and the remaining 50,000 iterations from all four chains are used for inference with no thinning. We checked the trace plots for convergence and ensured that the potential scale reduction $\hat{R} \leq 1.001$ and effective sample size per iteration is at least 0.1. All algorithms are run on a Intel Core i9-9900K CPU @ 3.60Ghz, 16.0 GB RAM. Lower bounds are estimated using 1000 simulations in each case and exclude constants independent of the variational parameters. The stopping criterion described in Section \ref{sec:Alg_imp} was used for RVB1, RVB2 and GVA.

\subsection{Simulations}
We simulate data from the Poisson, Bernoulli and binomial GLMMs to compare the performance of RVB1 and RVB2 in different scenarios, and identify the circumstances in which RVB2 is preferred to RVB1 despite it being computationally more intensive. Consider the random intercept model,
\begin{equation*}
\begin{aligned}
\eta_{ij} &= \beta_0 + \beta_1 x_{ij} + b_i, \quad b_i \sim \N(0, \sigma^2), \quad  \text{for} \quad i=1, \dots, n, \, j = 1, \,  \dots n_i.
\end{aligned}
\end{equation*}
We let $n =500$, $n_i = 7$ for each $i$ and $\sigma =1.5$. MCMC results reported in Table \ref{tab_sim} correspond to the centered parametrization.

For the Poisson model, we let $x_{ij} = (j-4)/10$ and consider $\beta_0 = -2.5$, $\beta_1 = -2$ for model I and $\beta_0 = 1.5$, $\beta_1 = 0.5$ for model II. Simulated responses from model I consist of small counts, with 82.2\% zeros and 99.3\% smaller or equal to 5. In contrast, responses for model II consist of large counts with 12.7\% zeros, 42.0\% in the interval $[1,5]$ and 45.3\% greater than 5. Results in Table \ref{tab_sim} indicate that RVB2 attains a higher lower bound and hence provides a more accurate posterior approximation than both RVB1 and GVA for model I. In particular, the posterior variance of the global parameters is captured more accurately by RVB2. For model II, RVB1 performs as well as RVB2 and both perform better than GVA. These observations are in line with the discussion in Section \ref{Sec_diff} that RVB2 works better than RVB1 for data with small counts. However, RVB2 is slower than RVB1 for both models even though it converges in the same or smaller number of iterations. GVA takes much more iterations to converge than RVB, especially for model II. For MCMC, centering is much faster than noncentering (484.8 s for model I and 840.9 s for model II) and has much higher effective sample sizes. INLA provides a fit closest to MCMC in both cases and is the fastest.

For the Bernoulli model, we set $\beta_0 = -2.5$, $\beta_2 = 4.5$ and randomly draw $x_{ij} \sim \text{Bernoulli}(0.5)$ for model I. For model II, $\beta_0 = 0$, $\beta_2 = 1$ and $x_{ij} = (j-4)/10$. This results in $\{\mu_{ij}\}$ being closer to 0 or 1 for model I, and more concentrated about 0.5 for model II (see Figure \ref{fig_mus}). Table \ref{tab_sim} shows that RVB2 produces a fit for the global parameters closest to MCMC for both models, among all other methods. As expected, RVB1 performs worse than RVB2 for model I and nearly as well for model II. For MCMC, noncentering takes more than twice as long (385.3 s for model I and 461.4 s for model II), although effective sample sizes are higher than centering for model I.

We use the same settings for the binomial model as we did for the Bernoulli model, and set $m_{ij} = 20$ uniformly. For model I, RVB1's performance is still weaker than other variational methods even though $m_{ij}$ is large. For model II, RVB1 performs as well as RVB2. All methods produce marginal posterior distributions for the global parameters that are very close to that of MCMC. RVB2 required the least number of iterations to converge for both models among the variational methods. For MCMC, noncentering required a shorter runtime of 939.7 s for model I and 1466.9 s for model II, but the effective sample size per iteration for $\beta_0$ is less than 0.1 for both models.

\begin{table}
\caption{ \label{tab_sim}
Simulated data: estimates of posterior mean and standard deviation for each parameter, runtime in seconds (number of iterations $(\times 10^3)$ in brackets) and lower bound.} 
\centering
\fbox{
\begin{tabular}{r|r|r|ccccc}
\hline
&  & & GVA & RVB1 & RVB2 & INLA & MCMC \\  \hline
\parbox[c]{1mm}{\multirow{10}{*}{\rotatebox[origin=c]{90}{Poisson}}} & \multirow{5}{*}{I} & $\beta_0$ & -2.46 $\pm$ 0.08 & -2.45 $\pm$ 0.08 & -2.45 $\pm$ 0.10 & -2.48 $\pm$ 0.10 & -2.51 $\pm$ 0.11 \\
&& $\beta_1$ & -2.26 $\pm$ 0.16 & -2.26 $\pm$ 0.15 & -2.26 $\pm$ 0.16 & -2.26 $\pm$ 0.16 & -2.26 $\pm$ 0.16 \\
& & $\sigma$ & 1.52 $\pm$ 0.06 & 1.51 $\pm$ 0.06 & 1.51 $\pm$ 0.08 & 1.55 $\pm$ 0.09 & 1.61 $\pm$ 0.09 \\
& & time & 27.7 (19) & 20.1 (15) & 27.6 (11) & 1.6 & 107.1 \\
& & $\hat{L}$ & -1316.9 & -1316.7 & -1316.3 & - & - \\   \cline{2-8}
& \multirow{5}{*}{II} & $\beta_0$ & 1.55 $\pm$ 0.06 & 1.55 $\pm$ 0.07 & 1.55 $\pm$ 0.07 & 1.55 $\pm$ 0.07 & 1.55 $\pm$ 0.07 \\
& & $\beta_1$  & 0.53 $\pm$ 0.02 & 0.53 $\pm$ 0.02 & 0.53 $\pm$ 0.02 & 0.53 $\pm$ 0.02 & 0.53 $\pm$ 0.02 \\
& & $\sigma$  & 1.47 $\pm$ 0.05 & 1.47 $\pm$ 0.05 & 1.47 $\pm$ 0.05 & 1.47 $\pm$ 0.05 & 1.47 $\pm$ 0.05 \\
& & time  & 128.7 (89) & 20.1 (15) & 33.5 (15) & 1.8 & 314.3 \\ 
& & $\hat{L}$ & 147130.5 & 147130.8 & 147130.8 & - & -  \\ \hline
\parbox[c]{1mm}{\multirow{10}{*}{\rotatebox[origin=c]{90}{Bernoulli}}}& \multirow{5}{*}{I} & $\beta_0$ & -2.54 $\pm$ 0.11 & -2.54 $\pm$ 0.11 & -2.54 $\pm$ 0.12 & -2.51 $\pm$ 0.12 & -2.55 $\pm$ 0.13 \\
& & $\beta_1$ & 4.59 $\pm$ 0.13 & 4.59 $\pm$ 0.12 & 4.59 $\pm$ 0.16 & 4.54 $\pm$ 0.16 & 4.62 $\pm$ 0.17 \\
& & $\sigma$ & 1.47 $\pm$ 0.07 & 1.47 $\pm$ 0.08 & 1.47 $\pm$ 0.10 & 1.40 $\pm$ 0.10 & 1.49 $\pm$ 0.10 \\
& & time & 29.8 (19) & 21.9 (15) & 58.5 (19) & 1.4 & 178.9 \\
& & $\hat{L}$ &  -1398.2 & -1398.2 & -1397.8 & - & - \\\cline{2-8}
& \multirow{5}{*}{II} &  $\beta_0$ & 0.07 $\pm$ 0.08 & 0.06 $\pm$ 0.08 & 0.07 $\pm$ 0.08 & 0.07 $\pm$ 0.08 & 0.07 $\pm$ 0.09 \\
& & $\beta_1$ & 1.35 $\pm$ 0.21 & 1.35 $\pm$ 0.21 & 1.35 $\pm$ 0.21 & 1.34 $\pm$ 0.20 & 1.35 $\pm$ 0.21 \\
& & $\sigma$ & 1.61 $\pm$ 0.07 & 1.60 $\pm$ 0.08 & 1.60 $\pm$ 0.08 & 1.58 $\pm$ 0.09 & 1.64 $\pm$ 0.09 \\
& & time & 33.7 (22) & 16.0 (11) & 23.2 (8) & 1.3 & 209.9 \\
& & $\hat{L}$ & -2127.3 & -2127.2 & -2127.0 & - & - \\ \hline
\parbox[c]{1mm}{\multirow{10}{*}{\rotatebox[origin=c]{90}{Binomial}}} & \multirow{5}{*}{I} & $\beta_0$ & -2.52 $\pm$ 0.07 & -2.51 $\pm$ 0.07 & -2.52 $\pm$ 0.07 & -2.52 $\pm$ 0.07 & -2.52 $\pm$ 0.07 \\
& & $\beta_1$ & 4.53 $\pm$ 0.03 & 4.53 $\pm$ 0.03 & 4.53 $\pm$ 0.03 & 4.53 $\pm$ 0.03 & 4.53 $\pm$ 0.03 \\
& & $\sigma$ & 1.49 $\pm$ 0.05 & 1.49 $\pm$ 0.05 & 1.49 $\pm$ 0.05 & 1.49 $\pm$ 0.05 & 1.49 $\pm$ 0.05 \\
& & time & 45.8 (28) & 42.7 (28) & 67.9 (19) & 1.6 & 1027.4 \\
& & $\hat{L}$ & -24301.3 & -24301.6 & -24301.3 & - & - \\  \cline{2-8}
& \multirow{5}{*}{II} &  $\beta_0$ & 0.08 $\pm$ 0.07 & 0.08 $\pm$ 0.07 & 0.08 $\pm$ 0.07 & 0.08 $\pm$ 0.07 & 0.08 $\pm$ 0.07 \\
& & $\beta_1$ & 1.09 $\pm$ 0.05 & 1.09 $\pm$ 0.05 & 1.09 $\pm$ 0.05 & 1.09 $\pm$ 0.05 & 1.09 $\pm$ 0.05 \\
& & $\sigma$ & 1.52 $\pm$ 0.05 & 1.52 $\pm$ 0.05 & 1.52 $\pm$ 0.05 & 1.52 $\pm$ 0.05 & 1.52 $\pm$ 0.05 \\
& & time & 44.2 (27) & 35.4 (24) & 37.3 (11) & 1.2 & 1212.4 \\
& & $\hat{L}$ &-37531.5 & -37531.4 & -37531.4 & - & - \\
\end{tabular}}
\end{table}
\begin{figure}
\centering
\makebox{\includegraphics[height=250pt, angle=270]{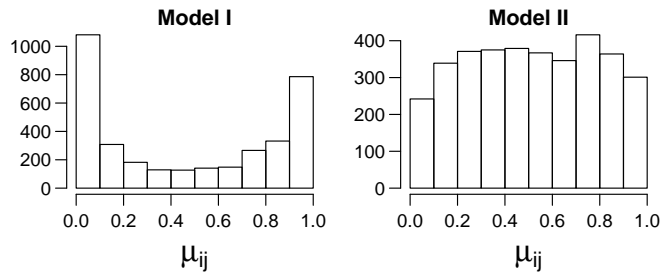}}
\caption{\label{fig_mus}Distribution of $\mu_{ij}$ for data simulated from Bernoulli and binomial models.}
\end{figure}

The simulations support recommendations in Section \ref{Sec_diff} that RVB1 is sufficient for count data without an excess of zeros, and binary data or binomial outcomes whose success probabilities are not concentrated at 0 or 1. Otherwise, RVB2 is preferred. In addition, we observe that RVB is able to yield improvements in estimating the posterior variance as compared to GVA, which has a tendency towards underestimation.

\subsection{Epilepsy data}
In the epilepsy data of \cite{Thall1990} ({\tt R} package {\tt MASS}, {\tt data(epil)}), $n=59$ epileptics were assigned either a new drug Progabide or a placebo randomly. The response $y_{ij}$ is the number of epileptic seizures of patient $i$ in the two weeks before clinic visit $j$ for $j=1, \dots, 4$. The covariates for patient $i$ are $\text{Base}_i$ (log of 1/4 the number of baseline seizures), $\text{Trt}_i$  (1 for drug treatment and 0 for placebo), $\text{Age}_i$ (log of age of patient at baseline, centered at zero), $\text{Visit}_{ij}$ (coded as $-0.3$, $-0.1$, $0.1$ and $0.3$ for $j=1, \dots, 4$ respectively) and V4, an indicator for the fourth visit. We consider the Poisson random intercept and slope models of \cite{Breslow1993}:
\begin{equation*}
\begin{aligned}
\text{Model I}&: \log \mu_{ij} = \beta_0+\beta_{\text{Base}} \text{Base}_i+\beta_{\text{Trt}} \text{Trt}_i + \beta_{\text{Base} \times \text{Trt}} \text{Base}_i \times \text{Trt}_i +\beta_{\text{Age}} \text{Age}_i  \\
& + \beta_{\text{V4}} \text{V4}_{ij} +b_i, \quad b_i \sim N(0, \sigma^2),  \\
\text{Model II}&: \log \mu_{ij} = \beta_0+\beta_{\text{Base}} \text{Base}_i+\beta_{\text{Trt}} \text{Trt}_i + \beta_{\text{Base} \times \text{Trt}} \text{Base}_i \times \text{Trt}_i +\beta_{\text{Age}} \text{Age}_i, \\
& +\beta_{\text{Visit}} \text{Visit}_{ij} +b_{i1} + b_{i2} \text{Visit}_{ij},
\quad \begin{bmatrix} b_{i1} \\ b_{i2} \end{bmatrix} \sim N(0, \Omega^{-1}), \,  \Omega^{-1} = \begin{bmatrix} \sigma_1^2 & \rho \sigma_1 \sigma_2 \\ \rho \sigma_1 \sigma_2 & \sigma_2^2 \end{bmatrix}.
\end{aligned}
\end{equation*}
For model I, the prior for $\sigma^{-2}$ is Gamma(0.5, 0.0151) and for model II, the prior for $\Omega$ is $W(3, S)$, where $S_{11} = 11.0169$, $S_{12} = -0.1616$ and $S_{22} = 0.5516$.

\begin{table}
\caption{ \label{tab_epil}
Epilepsy data. Estimates of posterior mean and standard deviation for each parameter, runtime in seconds (number of iterations $(\times 10^3)$ in brackets) and lower bound.} 
\centering
\fbox{\begin{tabular}{r|r|ccccc}
& & GVA & RVB1 & RVB2 & INLA & MCMC \\ \hline
\multirow{9}{*}{I} & $\beta_0$ & 0.27 $\pm$ 0.20 & 0.26 $\pm$ 0.27 & 0.27 $\pm$ 0.27 & 0.27 $\pm$ 0.27 & 0.26 $\pm$ 0.27 \\
& $\beta_{\text{Base}}$ & 0.88 $\pm$ 0.10 & 0.88 $\pm$ 0.13 & 0.88 $\pm$ 0.13 & 0.88 $\pm$ 0.14 & 0.89 $\pm$ 0.14 \\
& $\beta_{\text{Trt}}$ & -0.93 $\pm$ 0.40 & -0.94 $\pm$ 0.40 & -0.94 $\pm$ 0.41 & -0.94 $\pm$ 0.42 & -0.94 $\pm$ 0.42 \\
& $\beta_{\text{Base} \times \text{Trt}}$ & 0.33 $\pm$ 0.20 & 0.34 $\pm$ 0.21 & 0.34 $\pm$ 0.21 & 0.34 $\pm$ 0.21 & 0.34 $\pm$ 0.21 \\
& $\beta_{\text{Age}}$ & 0.48 $\pm$ 0.35 & 0.48 $\pm$ 0.36 & 0.47 $\pm$ 0.36 & 0.48 $\pm$ 0.36 & 0.48 $\pm$ 0.37 \\
& $\beta_{\text{V4}}$ & -0.16 $\pm$ 0.05 & -0.16 $\pm$ 0.05 & -0.16 $\pm$ 0.05 & -0.16 $\pm$ 0.05 & -0.16 $\pm$ 0.05 \\
& $\sigma$ & 0.52 $\pm$ 0.06 & 0.53 $\pm$ 0.06 & 0.53 $\pm$ 0.06 & 0.53 $\pm$ 0.06 & 0.53 $\pm$ 0.06 \\
& time & 15.2 (67) & 2.1 (10) & 3.1 (10) & 0.8 & 63.0 \\
& $\hat{\L}$ & 3130.7 & 3132.3 & 3132.4 & - & - \\ \hline
\multirow{9}{*}{II} & $\beta_0$ & 0.21 $\pm$ 0.21 & 0.21 $\pm$ 0.26 & 0.21 $\pm$ 0.26 & 0.21 $\pm$ 0.26 & 0.21 $\pm$ 0.27 \\
& $\beta_{\text{Base}}$ & 0.89 $\pm$ 0.10 & 0.89 $\pm$ 0.13 & 0.89 $\pm$ 0.13 & 0.89 $\pm$ 0.13 & 0.89 $\pm$ 0.14 \\
& $\beta_{\text{Trt}}$  & -0.93 $\pm$ 0.39 & -0.94 $\pm$ 0.40 & -0.94 $\pm$ 0.41 & -0.93 $\pm$ 0.41 & -0.93 $\pm$ 0.41 \\
& $\beta_{\text{Base} \times \text{Trt}}$ & 0.34 $\pm$ 0.20 & 0.34 $\pm$ 0.20 & 0.34 $\pm$ 0.20 & 0.34 $\pm$ 0.21 & 0.34 $\pm$ 0.21 \\
& $\beta_{\text{Age}}$ & 0.47 $\pm$ 0.34 & 0.48 $\pm$ 0.35 & 0.48 $\pm$ 0.36 & 0.48 $\pm$ 0.36 & 0.48 $\pm$ 0.36 \\
& $\beta_{\text{Visit}}$ & -0.26 $\pm$ 0.16 & -0.28 $\pm$ 0.16 & -0.28 $\pm$ 0.17 & -0.27 $\pm$ 0.17 & -0.27 $\pm$ 0.17 \\
& $\sigma_1$ & 0.51 $\pm$ 0.06 & 0.52 $\pm$ 0.06 & 0.52 $\pm$ 0.06 & 0.52 $\pm$ 0.06 & 0.52 $\pm$ 0.06 \\
& $\sigma_2$ & 0.77 $\pm$ 0.09 & 0.77 $\pm$ 0.14 & 0.77 $\pm$ 0.14 & 0.75 $\pm$ 0.14 & 0.76 $\pm$ 0.14 \\
& $\rho$ & 0.01 $\pm$ 0.17 & 0.01 $\pm$ 0.21 & 0.01 $\pm$ 0.22 & 0.01 $\pm$ 0.22 & 0.01 $\pm$ 0.23 \\
& time & 25.5 (70) & 4.8 (10) & 8.7 (10) & 1.2 & 144.9 \\
& $\hat{\L}$ & 3137.9 & 3140.1 & 3140.2 & - & - \\
\end{tabular}}
\end{table}

Counts in this data are relatively large, with only 9.7\% zeros and 48.8\% in [1,5]. Thus we expect RVB1 to perform as well as RVB2. From results in Table \ref{tab_epil}, RVB achieves higher lower bounds than GVA for both models and converges in a fraction of the time. RVB2 performs slightly better than RVB1. GVA captures the posterior means well but underestimates the posterior variance of several global parameters. Figure \ref{fig_epilLB} shows that RVB converges towards the mode of the lower bound much faster than GVA. In the first 1000 iterations, RVB has achieved a lower bound more than 3000 while GVA only reaches slightly over 2000. We observe again that for data with large counts, GVA converges very slowly, taking $\sim$7 times as many iterations as RVB. 
\begin{figure}
\centering
\makebox{\includegraphics[height=250pt, angle=270]{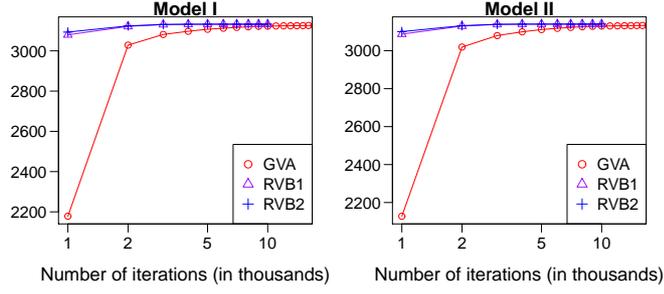}}
\caption{Epilepsy data. Average lower bound attained by GVA, RVB1, RVB2 every $10^3$ iterations.} \label{fig_epilLB}
\end{figure}
MCMC results are for centering; noncentering takes more than thrice as long (194.6 s for model I and 478.6 s for model II). Both RVB and INLA  produce very good posterior approximations of the global parameters, which are almost indistinguishable from MCMC.

Figure \ref{fig_epilbmeansd} shows that the posterior means of $\{\tilb_i\}$ estimated using RVB are somewhat close to zero while the standard deviations are close to one. This holds more strongly for RVB2 than RVB1 as Gaussian approximation of the conditional distributions of $\{b_i\}$ is more accurate in RVB2 (mode found using Newton-Raphson). 
\begin{figure}
\centering
\makebox{\includegraphics[height=400pt, angle=270]{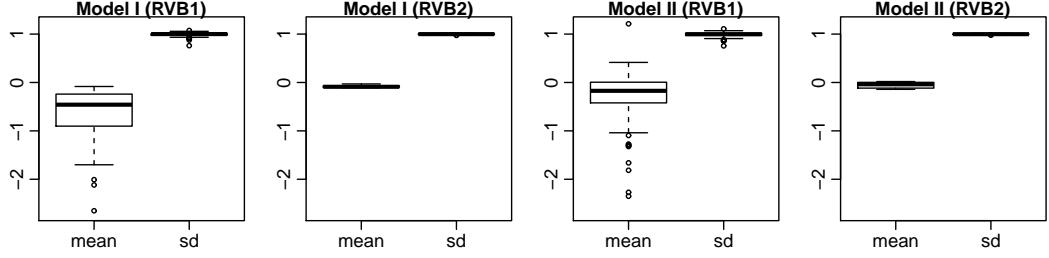}}
\caption{Epilepsy data. Boxplots of posterior means and standard deviations of $\{\tilb_i\}$.} \label{fig_epilbmeansd}
\end{figure}
These observations echo the proposition that $\{\tilb_i\}$ will be approximately uncorrelated a posteriori with zero means and unit variances after model reparametrization. Indeed, this property is very useful in improving the convergence of RVB as it makes initializing and optimizing variational parameters for $\{\tilb_i\}$ much easier. In high-dimensional problems with large number of local variables, computation time can be reduced by exploiting this feature.

Next we compare the marginal posteriors of random effects $\{b_i\}$ estimated using GVA and RVB with MCMC. For GVA, variational posterior of each $b_i$ is Gaussian, where the mean and variance can be obtained directly from the approximating density. For RVB, we use the approach in Section \ref{sec:b_i} to obtain 50000 samples of each $b_i$.
\begin{figure}
\centering
\makebox{\includegraphics[height=370pt, angle=270]{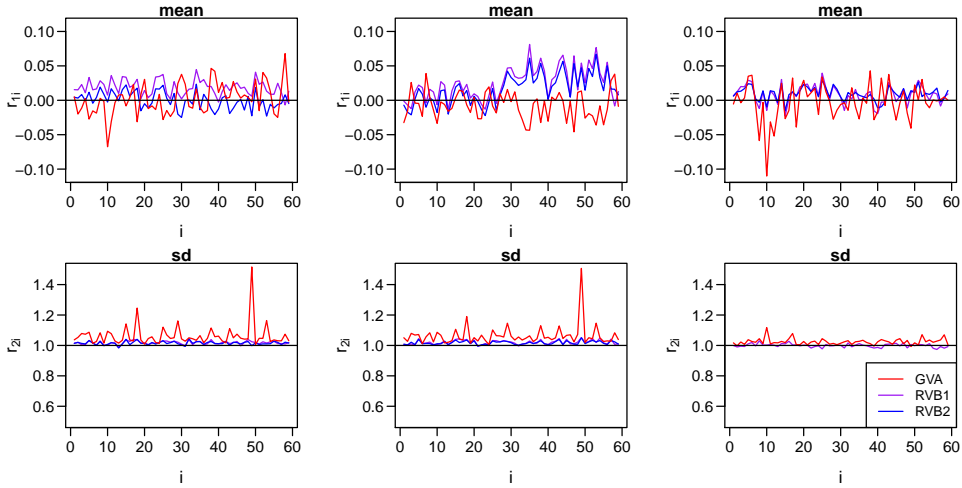}}
\caption{Epilepsy data. Plots of $r_{1i}$ and $r_{2i}$ against $i$ for GVA, RVB1 and RVB2. First column is for Model I, second and third columns are for $b_{i1}$ and $b_{i2}$ of Model II respectively.} \label{fig_epil_diff}
\end{figure}
Figure \ref{fig_epil_diff} plots 
\begin{equation} \label{rs}
r_{1i} = (\text{mean}_{b_i}^\text{VA} -\text{mean}_{b_i}^\text{MCMC})/\text{sd}_{b_i}^\text{VA} 
\quad \text{and} \quad 
r_{2i} = \text{sd}_{b_i}^\text{MCMC}/\text{sd}_{b_i}^\text{VA}
\end{equation}
where $\text{mean}_{b_i}^\text{method}$ and $\text{sd}_{b_i}^\text{method}$ denote the marginal posterior mean and standard deviation of $b_i$ estimated using a certain method. There are some instances where GVA underestimates the standard deviation quite severely while results of RVB are more uniform across all subjects and closer to MCMC generally. The posterior mean of $b_{i2}$ was severely underestimated for subject 10.

\subsection{Seeds data}
In the seeds germination data \citep{Crowder1978} ({\tt R} package {\tt hglm}, {\tt data(seeds)}), the response $y_i$ is the number of seeds that germinated out of $m_i$, which were brushed on plate $i$ for $i=1, \dots, 21$. This data arise from a $2^2$ factorial experiment and the factors are type of seed ($\text{seed}_i =1$ if O. aegyptica 73 and 0 if O. aegyptica 75) and type of root extract ($\text{extract}_i =1$ if cucumber and 0 if bean). We consider the binomial GLMM for handling overdispersion \citep{Breslow1993}, where $y_i \sim \text{binomial}(m_i, p_i)$,
\begin{equation*}
\text{logit} (p_i) = \beta_0 + \beta_{\text{seed}} \text{seed}_i + \beta_{\text{extract}} \text{extract}_i + b_i, \quad b_i \sim N(0, \sigma^2).
\end{equation*}
The prior for $\sigma^{-2}$ is Gamma$(0.5, 0.0544)$. 

\begin{figure}
\centering
\makebox{\includegraphics[height=370pt, angle=270]{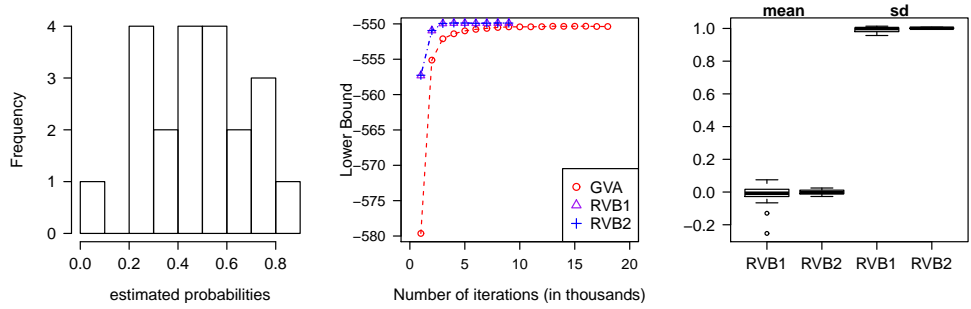}}
\caption{Seeds data. Left: histogram of $\{y_i/m_i\}$. Center: lower bound attained after every $10^3$ iterations. Right: boxplots of posterior means and standard deviations of $\{\tilb_i\}$.} \label{fig_seed_LB}
\end{figure}
A histogram of the raw probabilities $\{y_i/m_i\}$ (Figure \ref{fig_seed_LB}, left) does not indicate any concentration close to 0 or 1. Hence we postulate that RVB1 will perform as well as RVB2, and this is confirmed in Table \ref{tab_seed} (they achieve the same lower bound). RVB also achieves a higher lower bound than GVA, converging in half the number of iterations.  
\begin{table}
\caption{ \label{tab_seed}
Seeds data. Estimates of posterior mean and standard deviation for each parameter, runtime in seconds (number of iterations $(\times 10^3)$ in brackets) and lower bound.} 
\centering
\fbox{
\begin{tabular}{p{1.55cm}ccccc}
 & GVA & RVB1 & RVB2 & INLA & MCMC \\   \hline
$\beta_0$  & -0.38 $\pm$ 0.18 & -0.39 $\pm$ 0.18 & -0.39 $\pm$ 0.18 & -0.38 $\pm$ 0.19 & -0.38 $\pm$ 0.19 \\
$\beta_{\text{seed}}$ & -0.37 $\pm$ 0.23 & -0.36 $\pm$ 0.23 & -0.36 $\pm$ 0.23 & -0.37 $\pm$ 0.24 & -0.37 $\pm$ 0.24 \\
$\beta_{\text{extract}}$  & 1.04 $\pm$ 0.22 & 1.03 $\pm$ 0.22 & 1.03 $\pm$ 0.22 & 1.03 $\pm$ 0.23 & 1.03 $\pm$ 0.23 \\
$\sigma$ & 0.36 $\pm$ 0.07 & 0.35 $\pm$ 0.11 & 0.35 $\pm$ 0.11 & 0.36 $\pm$ 0.12 & 0.36 $\pm$ 0.12 \\
time  & 3.1 (18) & 1.6 (9) & 2.0 (9) & 0.5 & 19.4 \\
$\hat{\L}$ & -550.4 & -549.9 & -549.9 & - & - \\
\end{tabular}}
\end{table}
The center plot of Figure \ref{fig_seed_LB} shows a trend similar to that observed in the epilepsy data, where RVB converges to the mode much faster than GVA. The right plot shows that the means and standard deviations of $\{\tilb_i\}$ are concentrated around zero and one respectively (RVB2 more so than RVB1), which suggests that the affine transformation is effective in normalizing the random effects. Marginal posteriors of the global parameters estimated by INLA are closest to and indistinguishable from MCMC. MCMC results are from centering; noncentering takes 16.3 s but the effective sample size for $\sigma$ is lower than in centering. RVB also performs quite well while GVA underestimates the posterior variance of $\sigma$. Figure \ref{fig_seed_diff} summarizes the differences in mean and standard deviation of the marginal posteriors of $\{b_i\}$ between the variational methods and MCMC. RVB is generally able to capture the posterior means and standard deviations better than GVA although both methods underestimate the true standard deviation. 
\begin{figure}
\centering
\makebox{\includegraphics[height=300pt, angle=270]{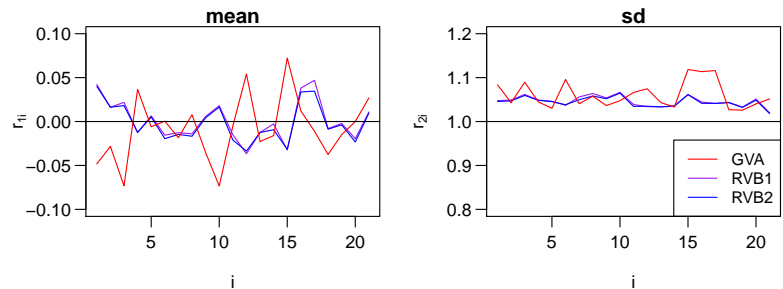}}
\caption{Seeds data. Plots of $r_{1i}$ and $r_{2i}$ against $i$ for GVA, RVB1 and RVB2.} \label{fig_seed_diff}
\end{figure}

\subsection{Toenail data}
This dataset ({\tt R} package {\tt HSAUR3}, {\tt data("toenail")}) contains the results of a clinical trial for comparing two oral antifungal treatments for toenail infection for 294 patients \citep{DeBacker1998}. The response $y_{ij}$ of patient $i$ at the $j$th clinic visit, $1 \leq j \leq 7$, is 1 if the degree of separation of the nail plate from the nail bed is moderate or severe and 0 if none or mild. For the $i$th patient, 250mg of terbinafine ($\text{Trt}_i =1$) or 200mg of itraconazole ($\text{Trt}_i =0$) is received per day, while $t_{ij}$ denotes the time in months when patient is evaluated at the $j$th visit. We consider the random intercept model,
\begin{equation*}
\text{logit} (p_{ij}) = \beta_0 + \beta_{\text{Trt}} \text{Trt}_i + \beta_t t_{ij} + \beta_{\text{Trt} \times t} \text{Trt}_i \times t_{ij} + b_i, \quad b_i \sim N(0, \sigma^2).
\end{equation*}
for $i=1, \dots, 294$, $1 \leq j \leq 7$. The prior for $\sigma^{-2}$ is Gamma(0.5, 0.4962). 

Results in Table \ref{tab_toenail} indicate that RVB2 achieved the highest lower bound, followed by GVA and RVB1.
\begin{table}
\caption{ \label{tab_toenail}
Toenail. Estimates of posterior mean and standard deviation for each parameter, runtime in seconds (number of iterations $(\times 10^3)$ in brackets) and lower bound.} 
\centering
\fbox{
\begin{tabular}{p{1.5cm}ccccc}
 & GVA & RVB1 & RVB2 & INLA & MCMC \\   \hline
$\beta_0$ & -3.22 $\pm$ 0.34 & -3.15 $\pm$ 0.31 & -3.23 $\pm$ 0.38 & -3.40 $\pm$ 0.41 & -3.51 $\pm$ 0.46 \\
$\beta_{\text{Trt}}$  & -0.76 $\pm$ 0.48 & -0.74 $\pm$ 0.45 & -0.75 $\pm$ 0.51 & -0.79 $\pm$ 0.50 & -0.82 $\pm$ 0.59 \\
$\beta_t$ & -1.64 $\pm$ 0.17 & -1.60 $\pm$ 0.14 & -1.64 $\pm$ 0.18 & -1.60 $\pm$ 0.17 & -1.71 $\pm$ 0.19 \\
$\beta_{\text{Trt} \times t}$ & -0.58 $\pm$ 0.27 & -0.54 $\pm$ 0.21 & -0.56 $\pm$ 0.27 & -0.56 $\pm$ 0.27 & -0.60 $\pm$ 0.29 \\
$\sigma$ & 3.54 $\pm$ 0.19 & 3.47 $\pm$ 0.16 & 3.56 $\pm$ 0.28 & 3.61 $\pm$ 0.32 & 4.10 $\pm$ 0.39 \\
time & 24.1 (29) & 16.0 (21) & 32.4 (15) & 1.1 & 361.7 \\
$\hat{\L}$ & -645.8 & -646.6 & -645.1 & - & - \\
\end{tabular}}
\end{table}
For binary data, it is difficult to estimate if the success probabilities are close to 0 or 1 and we recommend using RVB2 as default. RVB2 converged in the smallest number of iterations compared with RVB1 and GVA.
\begin{figure}
\centering
\makebox{\includegraphics[height=415pt, angle=270]{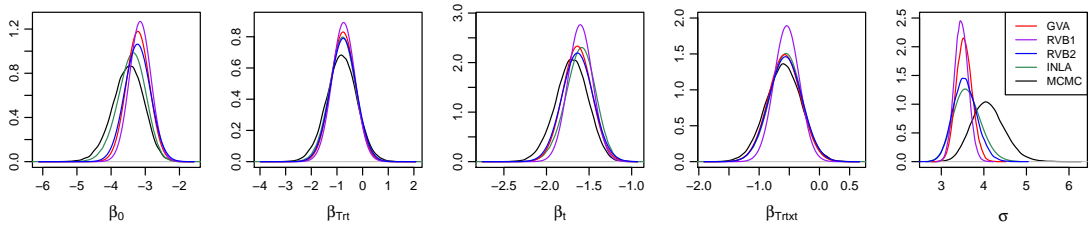}}
\caption{Toenail data. Plots of marginal posterior distributions of global parameters.} \label{fig_toenail}
\end{figure}
From Figure \ref{fig_toenail}, RVB2 performs as well as INLA for $\beta_{\text{Trt}}$, slightly better for $\beta_t$ and $\beta_{\text{Trt} \times t}$ and worse for $\beta_0$ and $\sigma$. However, it still improves significantly on GVA and RVB1. For MCMC, noncentering takes a longer time (404.3 s) than centering.

\subsection{HERS data}
We consider a large longitudinal data set from the Heart and Estrogen/Progestin Study \citep[HERS,][]{Hulley1998} available at \url{www.biostat.ucsf.edu/vgsm/data.html}. The study aims to determine if estrogen plus Progestin therapy reduces the risk of coronary heart disease (CHD) events for post-menopausal women with CHD. In this clinical trial, $2763$ women were randomly assigned to a hormone or placebo group and followed up for the next five years with an annual clinic visit. Some patients did not turn up for all 5 visits and data for certain covariates were missing. Here we consider 2031 women for whom data on all covariates are available. The binary response $y_{ij}$ is an indicator of whether the systolic blood pressure of patient $i$ is higher than 140 at the $j$th visit. We consider the random intercept model,
\begin{equation*}
\begin{aligned}
\text{logit} (p_{ij}) &= \beta_0 + \beta_{\text{age}} \text{age}_i +\beta_{\text{BMI}} \text{BMI}_{ij} +\beta_{\text{HTN}} \text{HTN}_{ij} + \beta_{\text{visit}} \text{visit}_{ij} + b_i, 
\end{aligned}
\end{equation*}
where $b_i \sim N(0, \sigma^2)$ for $i=1, \dots, 2031$, $0 \leq j \leq 5$. For patient $i$, $\text{age}_i$ is the age of patient at baseline, $\text{BMI}_{ij}$ is the body mass index at the $j$th visit, $\text{HTN}_{ij}$ is an indicator of whether high blood pressure medication is taken at the $j$th visit and $\text{visit}_{ij}$ is coded as $-1$, $-0.6$, $-0.2$, 0.2, 0.6, 1 for $j=0, 1, \dots, 5$ respectively. The covariates BMI and age are normalized before fitting the model. The prior for $\sigma^{-2}$ is Gamma(0.5, 0.5079).

\begin{table}
\caption{\label{tab_hers}
HERS. Estimates of posterior mean and standard deviation for each parameter, runtime in seconds (number of iterations $(\times 10^3)$ in brackets) and lower bound.} 
\centering
\fbox{
\begin{tabular}{p{1.3cm}ccccc}
 & GVA & RVB1 & RVB2 & INLA & MCMC \\   \hline
$\beta_0$ & -0.75 $\pm$ 0.08 & -0.75 $\pm$ 0.10 & -0.75 $\pm$ 0.10 & -0.77 $\pm$ 0.10 & -0.76 $\pm$ 0.11 \\
$\beta_{\text{age}}$ & 0.50 $\pm$ 0.05 & 0.50 $\pm$ 0.05 & 0.50 $\pm$ 0.05 & 0.50 $\pm$ 0.05 & 0.51 $\pm$ 0.06 \\
$\beta_{\text{BMI}}$ & 0.22 $\pm$ 0.05 & 0.21 $\pm$ 0.05 & 0.21 $\pm$ 0.05 & 0.22 $\pm$ 0.05 & 0.22 $\pm$ 0.05 \\
$\beta_{\text{HTN}}$ & -0.35 $\pm$ 0.09 & -0.35 $\pm$ 0.11 & -0.35 $\pm$ 0.11 & -0.34 $\pm$ 0.10 & -0.38 $\pm$ 0.11 \\
$\beta_{\text{visit}}$ & 0.22 $\pm$ 0.05 & 0.22 $\pm$ 0.05 & 0.23 $\pm$ 0.05 & 0.22 $\pm$ 0.05 & 0.23 $\pm$ 0.05 \\
$\sigma$ & 1.90 $\pm$ 0.04 & 1.89 $\pm$ 0.05 & 1.90 $\pm$ 0.06 & 1.85 $\pm$ 0.06 & 2.00 $\pm$ 0.07 \\
time & 429.5 (24) & 162.0 (11) & 240.0 (11) & 4.5 & 3415.7 \\
$\hat{\L}$ & -5041.9 & -5041.4 & -5041.1 & - & - \\
\end{tabular}}
\end{table}
\begin{figure}
\centering
\makebox{\includegraphics[height=415pt, angle=270]{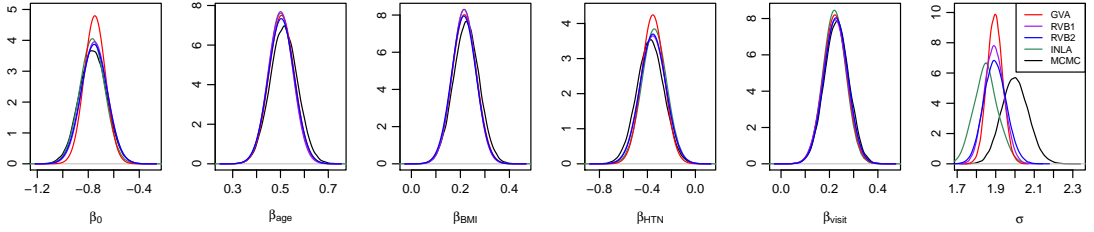}}
\caption{HERS data. Marginal posterior distributions of global parameters.} \label{fig_hers}
\end{figure}
From Table \ref{tab_hers}, RVB2 achieves the highest lower bound followed by RVB1 and GVA. RVB2 provides a very good fit to the marginal posteriors of the global parameters, which is better than INLA, although it still underestimates the mean and standard deviation of $\sigma$ (see Figure \ref{fig_hers}). The improvement it provides over GVA in estimating the posterior variance is evident. Both RVB1 and RVB2 converge in 11 iterations, which is less than half that taken by GVA. As a result, RVB2 is faster than GVA even though each iteration of RVB2 is computationally more intensive. Figure \ref{fig_hers_bmeansd} shows that GVA converges to the mode of the lower bound very slowly, unlike RVB1 and RVB2. For MCMC, centering (11000.2 s) takes more than thrice as long as noncentering and warnings on exceeding maximum treedepth were issued. 
\begin{figure}
\centering
\makebox{\includegraphics[height=300pt, angle=270]{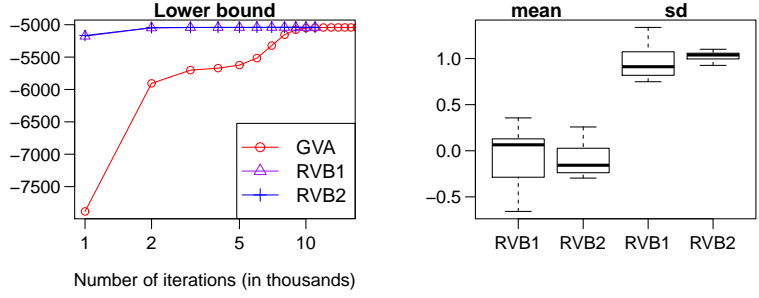}}
\caption{HERS. Lower bound and boxplots of posterior means and standard deviations of $\{\tilb_i\}$.} \label{fig_hers_bmeansd}
\end{figure}
The boxplots in Figure \ref{fig_hers_bmeansd} show that the means and variances of $\{\tilb_i\}$ are again close to zero and one respectively. However, the variation seems to be larger than previous examples, even for RVB2. This may suggest that approximating $p(b_i|\theta_G, y)$ by a Gaussian distribution is not as effective for logistic mixed models due to skewness or heavy tails. A direction for future work is to use skewed Gaussians or mixture of Gaussians to improve conditional posterior approximation.

We investigate the performance of the divide and recombine strategy by partitioning the subjects randomly into three groups, each with 677. RVB1 and RVB2 were applied to the three partial datasets in parallel and results are combined using techniques in Section \ref{sec:ext}. Here we use a N$(0,10^2)$ prior on $\omega$ instead of a Gamma prior on $\sigma^{-2}$ for greater ease in applying the divide and recombine strategy. This experiment was replicated ten times. Table \ref{tab_hers2} shows the results of RVB1 and RVB2 (with a normal prior on $\omega$) and corresponding results obtained using the divide and recombine strategy. The runtime is reduced on average by 5.9 times for RVB1 and 4.2 times for RVB2. We also compute an estimate of the lower bound attained in each experiment and the mean and standard deviation across the ten trials are reported. RVB2 performs better than RVB1 in this respect. Estimates of the global parameters obtained using the divide and recombine strategy are quite close to that obtained using the full algorithm and the standard deviation is less than or equal to 0.01 in all cases.  
\begin{table}
\caption{\label{tab_hers2}
HERS. Estimates of posterior mean and standard deviation for each parameter, runtime in seconds (number of iterations $(\times 10^3)$ in brackets) and lower bound. For the divide and recombine strategy, the means over ten replicates are reported with the standard deviation in brackets. A normal prior is used for $\omega$.} 
\centering
\fbox{
\begin{tabular}{r|cc|cc}
 & RVB1 & RVB1 (divide) & RVB2 & RVB2 (divide)   \\   \hline
$\beta_0$ & -0.75 $\pm$ 0.10 & -0.75(0.01) $\pm$ 0.10(0.00) & -0.75 $\pm$ 0.10 & -0.76(0.01) $\pm$ 0.10(0.00) \\
$\beta_{\text{age}}$ & 0.50 $\pm$ 0.05 & 0.50(0.00) $\pm$ 0.05(0.00) & 0.50 $\pm$ 0.05 & 0.50(0.00) $\pm$ 0.06(0.00) \\
$\beta_{\text{BMI}}$ & 0.22 $\pm$ 0.05 & 0.22(0.01) $\pm$ 0.05(0.00) & 0.21 $\pm$ 0.05 & 0.22(0.01) $\pm$ 0.05(0.00) \\
$\beta_{\text{HTN}}$ & -0.36 $\pm$ 0.11 & -0.36(0.01) $\pm$ 0.11(0.00) & -0.35 $\pm$ 0.11 & -0.35(0.01) $\pm$ 0.11(0.00) \\
$\beta_{\text{visit}}$ & 0.22 $\pm$ 0.05 & 0.22(0.00) $\pm$ 0.05(0.00) & 0.23 $\pm$ 0.05 & 0.23(0.00) $\pm$ 0.05(0.00) \\
$\omega$ & -0.64 $\pm$ 0.03 & -0.64(0.00) $\pm$ 0.03(0.00) & -0.64 $\pm$ 0.03 & -0.63(0.01) $\pm$ 0.03(0.00) \\
time & 162.1 & 27.3 (2.8) & 233.3 & 55.1 (6.1) \\
$\hat{\L}$  & -5040.7 & -5174.4 (2.2) & -5040.3 & -5095.0 (1.3) \\
\end{tabular}}
\end{table}

\section{Conclusion} \label{sec_conclusion}
We propose a model reparametrization approach to improve VB inference for hierarchical models, where the local variables are transformed via an affine transformation to minimize their posterior dependency on the global variables. The resulting Gaussian variational approximation (obtained using stochastic gradient ascent) is low-dimensional and can be readily extended to large datasets using a ``divide and recombine" strategy via parallel processing. In the application to GLMMs, we find that RVB1, which relies on data-based estimates of linear predictors, works well for Poisson models with large counts and binomial  outcomes where the success probabilities are not concentrated at 0 or 1. Binary data are more challenging to fit and RVB2, which searches for the conditional posterior modes of the random effects, works better for these models. The performance of MCMC in Stan relies heavily on the model parametrization both in terms of mixing and efficiency, and it is not always easy to find a parametrization that works well in both aspects. INLA is very fast and produces very good approximations of the marginal posteriors of the global parameters for Poisson and binomial models, but its performance is more varied for binary data and RVB2 can perform better sometimes. In addition, RVB provides an approximation of the full joint posterior distribution which can be very useful. Compared to GVA, RVB is often able to yield improvements in both convergence rate and accuracy of posterior approximation. The results obtained from GLMMs are overall very promising and it will be interesting to investigate the performance of RVB for more complex models.

\bibliographystyle{chicago}
\bibliography{ref}

\begin{center}
\Large \bf Supplementary material
\end{center}

\setcounter{section}{0} \renewcommand{\thesection}{S\arabic{section}}
\setcounter{figure}{0} \renewcommand{\thefigure}{S\arabic{figure}}
\setcounter{table}{0} \renewcommand{\thetable}{S\arabic{table}}
\setcounter{equation}{0} \renewcommand{\theequation}{S\arabic{equation}}
\setcounter{lemma}{0} \renewcommand{\thelemma}{S\arabic{lemma}}

\section{Default conjugate prior} \label{Sec_Wishart}
\cite{Kass2006} propose a default conjugate inverse Wishart prior, $IW(\rho, \rho R)$, for the $r \times r$ covariance matrix $\Omega^{-1}$, which is designed to be minimally informative by taking $\rho$ to be small, and the scale matrix $R$ is determined through first-stage data variability based on a prior guess for $\Omega^{-1}$. They begin by considering the GLM obtained by pooling all data and setting $b_i=0$ for all $i$. Let $\phi$ denote the dispersion parameter and $V(\mu_{ij})$ denote the variance function of the GLM. Suppose $\hat{\beta}$ is an estimate of the regression coefficients of the GLM. Define $W_i(\hat{\beta})$ as the $n_i \times n_i$ diagonal weight matrix of the GLM, whose $j$th diagonal entry is given by $\{\phi V(\mu_{ij}) [g'(\mu_{ij})]^2 \}^{-1}$ evaluated at $\mu_{ij}=\hat{\mu}_{ij} = g^{-1} (X_{ij}^T \hat{\beta})$. The authors suggested taking 
\begin{equation*}
\rho = r \quad \text{and} \quad R= \left( \frac{1}{n} \sum_{i=1}^n Z_i^T W_i(\hat{\beta}) Z_i \right)^{-1},
\end{equation*}
For the Poisson and binomial model, the dispersion parameter is one. The $j$th diagonal entry of $W_i(\hat{\beta})$ is given by 
$\hat{\mu} _{ij}= \exp(X_{ij}^T \hat{\beta})$ for the Poisson model and ${\hat{\mu}_{ij} (m_{ij} - \hat{\mu}_{ij})}/{m_{ij}} = {m_{ij}\exp(X_{ij}^T \hat{\beta})}{\{1 + \exp(X_{ij}^T \hat{\beta}) \}^{-2}}$ for the binomial model. See Table \ref{default} for details. An inverse Wishart prior $IW(\rho, \rho R)$ for $\Omega^{-1}$ is equivalent to a Wishart prior $W(\nu, S)$ for $\Omega$, where $\nu = \rho$ and $S = R^{-1}/\rho$. If $r=1$, the Wishart prior $W(\nu, S)$ reduces to $\text{Gamma}(\nu/2, S^{-1}/2)$. 

For applications where $r=2$, INLA crashed when we set $\nu=2$, even though RStan and the variational methods worked well. Hence we modify the default conjugate prior slightly by setting $\rho=r+1$ instead so that $\nu=r+1$ and $S = R^{-1}/(r+1)$ so that results are comparable with INLA.

\begin{table}
\caption{Estimate of the mean, variance functio, link function and derivative of the link function for computing the default conjugate prior. \label{default}}
\centering
\fbox{\begin{tabular}{l|cccc}
Model & $\hat{\mu}_{ij}$ & $V(\mu_{ij})$ & $g(\mu_{ij})$ & $g'(\mu_{ij})$ \\[1mm]  \hline
Poisson$(\mu_{ij})$ & $\exp(X_{ij}^T \hat{\beta})$ & $\mu_{ij}$ & $\log(\mu_{ij})$ & $1/\mu_{ij}$ \\ [1mm]
Binomial$\left( m_{ij}, \frac{\mu_{ij}}{m_{ij}} \right)$ &  $\frac{m_{ij}\exp(X_{ij}^T \hat{\beta})}{1 + \exp(X_{ij}^T \hat{\beta})}$  &  $\frac{\mu_{ij}(m_{ij} - \mu_{ij})}{m_{ij}} $ &  $\log\big(\frac{\mu_{ij}}{m_{ij} - \mu_{ij}}\big)$ & $\frac{m_{ij}}{\mu_{ij} (m_{ij} - \mu_{ij})}$ \\
\end{tabular}}
\end{table}

\section{Induced prior $ p(\omega)$}
Recall that $W_{ii} = \exp(W_{ii}^*)$ where $W_{ii} > 0$ and $W_{ii}^* \in \mathbb{R}$. Let $D^W$ be a diagonal matrix of order $r(r+1)/2$ where the diagonal is given by $\vech(J^W)$, and $J^W$ is an $r \times r$ matrix with $J^W_{ii} = W_{ii}$ and $J^W_{ij} = 1$ is $i \neq j$. Let $\d$ denote the differential operator \cite[see e.g.][]{Magnus1999}. Differentiating w.r.t. $\omega$, 
\begin{equation*}
\d \vec(W) = E_r^T \d \vech(W) = E_r^T D^W \d\omega.
\end{equation*}
Since $\Omega = WW^T$, differentiating w.r.t. $\omega$, 
\begin{equation} \label{dvecOmega}
\begin{aligned}
\d\vec(\Omega) & =  \vec\{ (\d W) W^T + W (\d W)^T\} \\
& =  (W \otimes I) \d \vec(W) + (I \otimes W) K_r \d \vec(W) \\
& = 2N_r (W \otimes I) \d \vec(W) \\
& = 2 N_r (W \otimes I) E_r^T D^W \d\omega. \\
\therefore \frac{\partial \vec(\Omega)}{\partial \omega} &= 2D^W E_r (W^T \otimes I) N_r.
\end{aligned}
\end{equation}

We have
\begin{equation*}
p(\omega) = p(\Omega) \bigg|\frac{\partial \vech(\Omega)}{\partial \omega} \bigg|. 
\end{equation*}
As $\vech(\Omega) = E_r \vec(\Omega)$, $\d \vech(\Omega) = E_r \d \vec(\Omega) = 2 E_r N_r (W \otimes I) E_r^T D^W \d\omega$ from \eqref{dvecOmega}. Hence 
\begin{equation*}
\begin{aligned}
\frac{\partial \vech(\Omega)}{\partial \omega} &= 2D^W E_r (W^T \otimes I) N_r E_r^T
\end{aligned}
\end{equation*}
and
\begin{equation*}
\begin{aligned}
\bigg|\frac{\partial \vech(\Omega)}{\partial \omega} \bigg| = 2^{r(r+1)/2}|D^W| \, |E_r (W^T \otimes I) N_r E_r^T|.
\end{aligned}
\end{equation*}
We have $|D^W| = \prod_{i=1}^r W_{ii}$ and 
\begin{equation*}
\begin{aligned}
|E_r (W^T \otimes I) N_r E_r^T| & =  |E_r (W^T \otimes I) D_r E_r N_r E_r^T| \\
& = |E_r (W^T \otimes I) D_r| \,|E_r N_r E_r^T| \\
& = 2^{-r(r-1)/2} \prod_{i=1}^r W_{ii}^{r-i+1}.
\end{aligned}
\end{equation*}
In evaluating the above expression, we have made use of the following results from \cite{Magnus1980}.
\begin{itemize}
\item Lemma 3.5 (ii): $N_r =  D_r E_r N_r$
\item Lemma 3.4 (i): $| E_r N_r E_r^T| = 2^{-r(r-1)/2}$.
\item Lemma 4.1 (iii): $|E_r (W^T \otimes I) D_r|  = \prod_{i=1}^r W_{ii}^{r-i+1}$.
\end{itemize}
Hence 
\begin{equation*}
\log p(\omega) = \log p(\Omega) + r \log 2 + \sum_{i=1}^r (r-i+2) \log (W_{ii}).
\end{equation*}

\section{Proof of Lemma 1}

\begin{proof}
If $\nexists$ $\eta_{ij} \in \mathbbm{R}$ such that $p(y_{ij}|\eta_{ij})$ attains its maximum value, then $\nexists$ $\eta_{ij} \in \mathbbm{R}$ such that $h_{ij}'(\eta_{ij}) = y_{ij}$. Since $h_{ij}'(\eta_{ij})$ is a continuous monotone increasing function ($h_{ij}''(\eta_{ij}) = \var(y_{ij}) \geq 0$), $h_{ij}'(\eta_{ij})$ must be either bounded below or bounded above by $y_{ij}$. If $h_{ij}'(\eta_{ij})$ is bounded below by $y_{ij}$, then $\exists$ $y_{ij} \leq L < \infty $ such that $\lim_{\eta_{ij} \rightarrow -\infty} h_{ij}'(\eta_{ij}) = L$. Otherwise, $\exists$ $-\infty < L \leq y_{ij}$ such that $\lim_{\eta_{ij} \rightarrow \infty} h_{ij}'(\eta_{ij}) = L$. Thus we have $\lim_{\eta_{ij} \rightarrow c} h_{ij}'(\eta_{ij}) = L$, where $c=-\infty$ or $c=\infty$. By L'Hospital's rule ($\lim_{x \rightarrow c} f(x)/g(x) = \lim_{x \rightarrow c} f'(x)/g'(x)$ if $\lim_{x \rightarrow c}|g(x)| = \infty$), 
\begin{equation*}
\begin{gathered}
0 = \lim_{\eta_{ij} \rightarrow c} \frac{h_{ij}'(\eta_{ij})}{\eta_{ij}} = \lim_{\eta_{ij} \rightarrow c} h_{ij}''(\eta_{ij}), \\
L = \lim_{\eta_{ij} \rightarrow c} \frac{h_{ij}'(\eta_{ij}) \eta_{ij}}{\eta_{ij}}  =  \lim_{\eta_{ij} \rightarrow c} \{h_{ij}''(\eta_{ij}) \eta_{ij} + h_{ij}'(\eta_{ij}) \}.
\end{gathered}
\end{equation*}
Hence, $\lim_{\eta_{ij} \rightarrow c} h_{ij}''(\eta_{ij}) = 0$ and $\lim_{\eta_{ij} \rightarrow c}  h_{ij}''(\eta_{ij}) \eta_{ij} = 0$. Since $\var(y_{ij}) \rightarrow 0$ as $\eta_{ij} \rightarrow c$, $p(y_{ij}|\eta_{ij})$ approaches a degenerate distribution with support only on $L$ ($\because$ $E(y_{ij}) \rightarrow L$). Since $p(y_{ij}|\eta_{ij})$ is concave down and the likelihood of observing $y_{ij}$ is maximized as $\eta_{ij} \rightarrow c$, $L = y_{ij}$ and $\lim_{\heta_{ij} \rightarrow c} h_{ij}'(\eta_{ij}) = y_{ij}$.
\end{proof}

\section{Proof of Theorem 1}
\begin{proof}
First, $\Lambda_i$ is well-defined for any value of $\heta_i$ since $h''(\heta_{ij}) \geq 0$. The $j$th element of $\hg_i   + \hH_i(\heta_i - X_i \beta)$ is 
\begin{equation} \label{jelement}
y_{ij} - h'(\heta_{ij}) + h_{ij}''(\heta_{ij}) \heta_{ij} - h_{ij}''(\heta_{ij}) X_{ij}^T \beta,
\end{equation}
which is well-defined if $\heta_{ij}^\ML$ exist. If $\heta_{ij}^\ML$ does not exist, then from Lemma \ref{lemA}, \eqref{jelement} has a limit of zero as $\heta_{ij} \rightarrow c_{ij}$. Thus the limit of $\lambda_{i}$ exists as $\heta_i \rightarrow c_i$. If $\heta_{ij}^\ML$ does not exist for $j=1, \dots, n_i$, then $\lim_{\heta_i \rightarrow c_i} \{\hg_i   + \hH_i(\heta_i - X_i \beta) \}= 0$ which implies that $\lim_{\heta_i \rightarrow c_i} \lambda_i = 0$.
\end{proof}

\section{Unbiased estimates of stochastic gradients}\label{apx:stocgrad}
As $\tiltheta = Cs + \mu$, differentiating $\ell(\tiltheta)$ with respect to $\mu$ and $\vech(C)$ separately, 
\begin{equation*}
\begin{aligned}[t]
\d \{ \ell(\tiltheta) \} &= \nabla_{\tiltheta} \ell(\tiltheta)^T \d\mu,
\end{aligned}
\qquad
\begin{aligned}[t]
\d \{ \ell(\tiltheta) \} &= \nabla_{\tiltheta} \ell(\tiltheta)^T (\d C) s \\
&= (s^T \otimes  \nabla_{\tiltheta} \ell(\tiltheta)^T)\d\vec(C) \\
&= \vec(\nabla_{\tiltheta} \ell(\tiltheta) s^T)^T E_d^T \d\vech(C) \\
&= \vech(\nabla_{\tiltheta} \ell(\tiltheta) s^T)^T \d\vech(C).
\end{aligned}
\end{equation*}
Therefore $\nabla _\mu \ell(\tiltheta) = \nabla_{\tiltheta} \ell(\tiltheta)$ and $\nabla _{\vech(C)} \ell(\tiltheta) = \vech(\nabla_{\tiltheta} \ell(\tiltheta) s^T)$. In addition,
\begin{equation*}
\d \log |C| = \tr(C^{-1} \d C) = \vec(C^{-T})^T E_d^T \d\vech(C) = \vech(C^{-T})^T  \d\vech(C).
\end{equation*}
Hence $\nabla_{\vech(C)} \log |C| = \vech(C^{-T})$. For the second estimator, differentiating $\log q(\tiltheta|\mu, C)$ with respect to $\tiltheta$,
\begin{equation*}
\begin{aligned}
\d \log q(\tiltheta|\mu, C) = -(\tiltheta - \mu)^T C^{-T} C^{-1} \d\tiltheta = - s^T C^{-1} \d\tiltheta,
\end{aligned}
\end{equation*}
Hence $\nabla_{\tiltheta} \log q(\tiltheta|\mu, C) = -C^{-T}s$. Differentiating with respect to $C$,
\begin{equation*}
\begin{aligned}
 - \d \{(\tiltheta - \mu)^T C^{-T} C^{-1} (\tiltheta - \mu)\} &= s^T \{(\d C^T) C^{-T} + C^{-1} (\d C) \} s \\
&=\{ (s^T C^{-1} \otimes s^T) K_d + (s^T \otimes s^T C^{-1})\} \d\vec(C) \\
&=\{ \vec(ss^T C^{-1} )^T K_d + \vec(C^{-T}ss^T)^T\} \d\vec(C) \\
&=2\vec(C^{-T}ss^T)^TE_d^T \d\vech(C) \\
&=2\vech(C^{-T}ss^T)^T \d\vech(C) \\
\end{aligned}
\end{equation*}
Hence $\nabla_{\vech(C)} \{ \log q(\tiltheta|\mu, C)  \} = \vech( C^{-T}ss^T-C^{-T} )$.

\section{Proof of Theorem 2}
 \begin{definition}
For any square matrix $A$, let $k(\cdot)$ denote a function such that 
$$k(A) = \bar{A} -  \tfrac{1}{2}\dg(A).$$
\end{definition}
Recall that $\dg(A)$ denotes the diagonal matrix derived from $A$ by setting non-diagonal elements to zero and $\bar{A}$ denotes the lower triangular matrix derived from $A$ by setting all superdiagonal elements to zero. 

\begin{lemma} \label{lem1}
Consider the differential with respect to $\theta_G$. Let $A_i = L_i^{-1} \d \Lambda_i L_i^{-T}$ for $i=1, \dots, n$. Then
\begin{enumerate}
\itemsep 0em
\item $\d L_i = L_i k( A_i )$,
\item $\vec\{k(A_i)\} = \tfrac{1}{2}  E_r^T D_r^T(L_i^{-1} \otimes L_i^{-1} ) \d\vec(\Lambda_i)$,
\item $\tr(L_i^{-1} \d L_i)  =  \tfrac{1}{2} \vec(\Lambda_i^{-1})^T \d\vec(\Lambda_i)$,
\item $a_i^T \d b_i = \tfrac{1}{2}\vec (L_i^{-T} \tilde{B}_i L_i^{-1})^T \d \vec(\Lambda_i)  + a_i^T\d\lambda_i$.
\end{enumerate}
\end{lemma}
\begin{proof}
For (a), to find the differential of the Cholesky factor $L_i$, we differentiate $L_i L_i^T = \Lambda_i$ and then multiply by $L_i^{-1}$ on the left and $L_i^{-T}$ on the right \citep{Murray2016}:
\begin{equation*}
\begin{gathered}
(\d L_i) L_i^T + L_i (\d L_i)^T = \d \Lambda_i, \\
L_i^{-1} \d L_i + (\d L_i)^T L_i^{-T} = L_i^{-1} \d \Lambda_i L_i^{-T}.
\end{gathered}
\end{equation*}
 On the left-hand side, the first term is lower triangular and the second term is upper triangular (transpose of the first term). The term on the right-hand side is $A_i$. Thus, $L_i^{-1} \d L_i = k( A_i )$ which implies that $\d L_i = L_i k( A_i )$.

For (b), we make use of some results from \cite{Magnus1980}:
\begin{itemize}
\item Lemma 3.3 (i): $\vec(\bar{A_i}) = E_r^T E_r \vec(A_i)$ 
\item Lemma 3.3 (iv): $\vec(\dg(A_i)) = E_r^T E_r K_r E_r^T E_r \vec(A_i)$.
\item Lemma 3.4 (ii) and Lemma 3.6 (iii): $2I_{r(r+1)/2} - E_r K_r E_r^T = D_r^T D_r $.
\item Lemma 4.4 (i): $D_r E_r(L_i^{-1} \otimes L_i^{-1} ) D_r = (L_i^{-1} \otimes L_i^{-1} ) D_r$.
\end{itemize}
Combining these results with the Kronecker product property that $\vec(A_i) = (L_i^{-1} \otimes L_i^{-1}) \d\vec(\Lambda_i)$, we have
\begin{equation*}
\begin{aligned}
\vec\{k(A_i)\} 
&= \vec(\bar{A_i}) - \tfrac{1}{2}\vec\{ \dg(A_i) \} \\
&= E_r^T E_r \vec(A_i) - \tfrac{1}{2}E_r^T E_r K_r E_r^T E_r \vec(A_i)  \\
&= \tfrac{1}{2} E_r^T (2I_{r(r+1)/2}  - E_r K_r E_r^T) E_r \vec(A_i) \\
&= \tfrac{1}{2} E_r^T D_r^T D_r  E_r  (L_i^{-1} \otimes L_i^{-1} )\d\vec(\Lambda_i) \\
&= \tfrac{1}{2} E_r^T D_r^T D_r E_r(L_i^{-1} \otimes L_i^{-1} ) D_r \d\vech(\Lambda_i) \\
&= \tfrac{1}{2} E_r^T D_r^T(L_i^{-1} \otimes L_i^{-1} ) \d\vec(\Lambda_i).
\end{aligned}
\end{equation*}

For (c), we can use the result from Lemma \ref{lem1} (a) to obtain
\begin{equation*}
\tr(L_i^{-1} \d L_i) = \tr(k(A_i)) =  \tfrac{1}{2} \tr(A_i)  =  \tfrac{1}{2} \tr(L_i^{-T} L_i^{-1} \d \Lambda_i) =  \tfrac{1}{2} \vec(\Lambda_i^{-1})^T \d\vec(\Lambda_i).
\end{equation*}

For (d), differentiating $b_i = L_i \tilb_i + \lambda_i$ with respect to $\theta_{G}$, we have $\d b_i = (\d L_i) \tilb_i + \d\lambda_i$. Hence $a_i^T \d b_i = a_i^T(\d L_i) \tilb_i + a_i^T\d\lambda_i$. Note that $(u \otimes v) = \vec(vu^T)$ if $u$ and $v$ are vectors and recall that $B_i = L_i^T a_i \tilde{b}_i^T$, $\tilde{B}_i = \bar{B}_i + \bar{B}_i^T - \dg(B_i)$. From Lemma \ref{lem1} (a) and (b), 
\begin{equation*}
\begin{aligned}
a_i^T (\d L_i) \tilb_i = a_i^TL_i k(A_i) \tilb_i  
&= (\tilb_i \otimes L_i^T a_i)^T \vec(k(A_i)) \\
&= \tfrac{1}{2} \vec(B_i)^T E_r^T D_r^T(L_i^{-1} \otimes L_i^{-1} ) \d\vec(\Lambda_i) \\
&= \tfrac{1}{2} \{ (L_i^{-T} \otimes L_i^{-T}) D_r \vech(B_i) \}^T \d\vec(\Lambda_i) \\
&= \tfrac{1}{2} \vec (L_i^{-T} \tilde{B}_i L_i^{-1})^T \d \vec(\Lambda_i).
\end{aligned}
\end{equation*}
\end{proof}

\begin{proof}[Theorem 2]
Note that $b_i = L_i \tilb_i + \lambda_i$. Differentiating $\ell(\tiltheta)$ with respect to $\tilb_i$, we obtain 
\begin{equation*}
\begin{aligned}
\d \ell(\tiltheta) &= \{\nabla_{b_i} \log p(y_i, b_i|\theta_G)\}^T \d b_i \\
&= a_i^T L_i \d\tilb_i \\
\implies \nabla_{\tilb_i}\ell(\tiltheta) &= L_i^T a_i.
\end{aligned}
\end{equation*}

Differentiating $\ell(\tiltheta)$ with respect to $\theta_{G_m}$ for $m=1, \dots, M$, 
\begin{equation*}
\begin{aligned}
\d \ell(\tiltheta) &= \sum_{i=1}^n [a_i^T \d b_i + \{\nabla_{\theta_{G_m}} \log p(y_i, b_i|\theta_G)\}^T \d\theta_{G_m} + \tr(L_i^{-1} \d L_i)] \\
& \quad + \{\nabla_{\theta_{G_m}} \log p(\theta_G) \}^T \d\theta_{G_m}.
\end{aligned}
\end{equation*}
The first two terms are obtained using chain rule as $ \log p(y_i, b_i|\theta_G)$ depends on $\theta_G$ directly as well as through $b_i$. Substituting the expressions of $\tr(L_i^{-1} \d L_i)$ and $a_i^T \d b_i$ from Lemma \ref{lem1} (c) and (d) respectively, we obtain 
\begin{equation*}
\begin{aligned}
\d \ell(\tiltheta) &= \sum_{i=1}^n [a_i^T\d\lambda_i  + \tfrac{1}{2}\vec (L_i^{-T} \tilde{B}_i L_i^{-1})^T \d \vec(\Lambda_i)  + \tfrac{1}{2} \vec(\Lambda_i^{-1})^T \d\vec(\Lambda_i)] \\
& \quad + \sum_{i=1}^n \{\nabla_{\theta_{G_m}} \log p(y_i, b_i|\theta_G)\}^T \d\theta_{G_m} + \{\nabla_{\theta_{G_m}} \log p(\theta_G) \}^T \d\theta_{G_m}.\\
\therefore \nabla_{\theta_{G_m}} \ell(\tiltheta) &=  \sum_{i=1}^n [ (\nabla_{\theta_{G_m}} \lambda_i) a_i + \tfrac{1}{2} \{\nabla_{\theta_{G_m}} \vec(\Lambda_i)\}  \vec (L_i^{-T} \tilde{B}_i L_i^{-1} + \Lambda_i^{-1})] \\
& \quad + \sum_{i=1}^n \{\nabla_{\theta_{G_m}} \log p(y_i, b_i|\theta_G)\} + \{\nabla_{\theta_{G_m}} \log p(\theta_G) \}.
\end{aligned}
\end{equation*}
\end{proof}

\section{Proof of Lemma 2}
\begin{proof}
We differentiate each expression below with respect to $\omega$. 
For (a),  
\begin{equation*}
\begin{aligned}
\d  \log p(y_i, b_i|\theta_G) 
&= \tfrac{1}{2} \d [\log |\Omega| - b_i^T \Omega b_i] \\
& = \tfrac{1}{2} [ \tr(\Omega^{-1} \d\Omega) - \tr(b_i b_i^T \d\Omega) ] \\
& = \tfrac{1}{2} [ \vec(\Omega^{-1})  - \vec(b_i b_i^T) ]^T\d \vec(\Omega) \\
& = \vec(\Omega^{-1} - b_i b_i^T)^T N_r (W \otimes I) E_r^T D^W \d \omega \\
\therefore \nabla_{\omega} \log p(y_i, b_i|\theta_G)  
& = D^W E_r (W^T \otimes I) \vec(\Omega^{-1} - b_i b_i^T) \\
& = D^W E_r \vec(\Omega^{-1}W - b_i b_i^TW) \\
& = D^W \vech(W^{-T} - b_i b_i^TW).
\end{aligned}
\end{equation*}

For (b), recall that 
\begin{equation*}
\log p(\omega) = \tfrac{\nu-r-1}{2} \log|\Omega| - \tfrac{1}{2} \tr(S^{-1}\Omega) + \sum_{i=1}^r  (r-i+2) \log (W_{ii}) + C_3,
\end{equation*}
where $C_3$ is a constant independent of $\omega$. Hence
\begin{equation*}
\begin{aligned}
\d \log p(\theta_G) 
&= \tfrac{\nu-r-1}{2} \tr(\Omega^{-1}\d\Omega) - \tfrac{1}{2} \tr(S^{-1}\d\Omega) + \sum_{i=1}^r (r-i+2) \d W_{ii}^* \\
& = \vec \{  (\nu-r-1) \Omega^{-1} - S^{-1} \} N_r (W \otimes I) E_r^T D^W \d\omega + \vech(\diag(u))^T \d\omega\\
\end{aligned}
\end{equation*}
where $u$ is a vector of length $r$ where the $i$th element is $r-i+2$.
\begin{equation*}
\begin{aligned}
\nabla_\omega  \log p(\theta_G)  
&= D^W  E_r (W^T \otimes I)\vec \{  (\nu-r-1) \Omega^{-1} - S^{-1} \} + \vech(\diag(u))\\
&= D^W  \vech \{(\nu-r-1) W^{-T} - S^{-1} W\} + \vech(\diag(u)).
\end{aligned}
\end{equation*}
\end{proof}

\section{Proof of Lemma 3}
\begin{proof}
In the first approach, we have 
\begin{equation*}
\Lambda_i = (\Omega + Z_i^T H_i(\hat{\eta}_i) Z_i)^{-1}, \quad 
\lambda_i = \Lambda_i  Z_i^T \{ y_i - g_i(\hat{\eta}_i) + H_i(\hat{\eta}_i)(\heta_i - X_i \beta) \}.
\end{equation*}

Differentiating each of the above expressions with respect to $\beta$, $\nabla_\beta \vec(\Lambda_i) = 0$ since $\Lambda_i$ is independent of $\beta$ and 
\begin{equation*}
\begin{aligned}
\d \lambda_i = - \Lambda_i  Z_i^T H_i(\hat{\eta}_i) X_i \d \beta
\implies \nabla_\beta  \lambda_i  = - X_i^T H_i(\hat{\eta}_i) Z_i \Lambda_i.
\end{aligned}
\end{equation*}
We also have $\nabla_{\beta} \log p(\theta_G) =  - \beta/ \sigma_\beta^2$ and $\nabla_{\beta} \log p(y_i, b_i|\theta_G) = X_i^T (y_i - g_i(\eta_i) )$. Plugging these into Theorem \ref{thm2}
\begin{equation*}
\begin{aligned}
\nabla_\beta \ell(\tiltheta) &= - \sum_{i=1}^n  X_i^T H_i(\hat{\eta}_i) Z_i \Lambda_i a_i + \sum_{i=1}^n X_i^T (y_i - g_i(\eta_i) ) - \beta/ \sigma_\beta^2 \\
&= \sum_{i=1}^n X_i^T (y_i - g_i(\eta_i) - H_i(\hat{\eta}_i) Z_i \Lambda_i a_i ) - \beta/ \sigma_\beta^2.
\end{aligned}
\end{equation*}

Differentiating $\lambda_i$ and $\Lambda_i$ with respect to $\omega$, we have 
\begin{equation*}
\begin{aligned}
\d \vec(\Lambda_i) 
& = - \vec \{\Lambda_i  (\d \Omega) \Lambda_i\} \\
&= - (\Lambda_i \otimes \Lambda_i) \d\vec(\Omega) \\
& = - 2 (\Lambda_i \otimes \Lambda_i) N_r (W \otimes I)  E_r^T D^W \d\omega \\
& = - 2 N_r (\Lambda_iW \otimes \Lambda_i) E_r^T D^W \d\omega. \\
\therefore \nabla_{\omega} \vec(\Lambda_i) &= - 2 D^W E_r (W^T\Lambda_i \otimes \Lambda_i) N_r.
\end{aligned}
\end{equation*}
In addition, 
\begin{equation*}
\begin{aligned}
\d \lambda_i  
&=  (\d\Lambda_i) Z_i^T \{ y_i - g_i(\hat{\eta}_i) + H_i(\hat{\eta}_i)(\heta_i - X_i \beta) \} \\
&=  - \Lambda_i (\d \Omega) \Lambda_i Z_i^T \{ y_i - g_i(\hat{\eta}_i) + H_i(\hat{\eta}_i)(\heta_i - X_i \beta) \} \\
&= - \Lambda_i  (\d \Omega) \lambda_i \\
& = - (\lambda_i^T \otimes \Lambda_i) 2N_r (W \otimes I) E_r^T D^W \d\omega \\
& = - \{(\lambda_i^T \otimes \Lambda_i)K_r + (\lambda_i^T \otimes \Lambda_i) \} (W \otimes I) E_r^T D^W \d\omega \\
& = - (\Lambda_i \otimes \lambda_i^T + \lambda_i^T \otimes \Lambda_i) (W \otimes I) E_r^T D^W \d\omega \\
& = - ( \Lambda_i W \otimes \lambda_i^T + \lambda_i^TW \otimes \Lambda_i) E_r^T D^W \d\omega. \\
\therefore \nabla_{\omega} \lambda_i &= - D^W E_r( W^{T} \Lambda_i \otimes \lambda_i + W^T\lambda_i \otimes \Lambda_i).
\end{aligned}
\end{equation*}
Plugging these results and Lemma \ref{lem2} into Theorem \ref{thm2},
\begin{equation*}
\begin{aligned}
\nabla_\omega \ell(\tiltheta) 
&= D^W \sum_{i=1}^n \vech(W^{-T} - b_ib_i^T W) - D^W E_r \sum_{i=1}^n (W^T \Lambda_i \otimes \lambda_i + W^T\lambda_i \otimes \Lambda_i)a_i \\
& \quad -D^W E_r \sum_{i=1}^n (W^T\Lambda_i \otimes \Lambda_i) \vec (\Lambda_i^{-1} + L_i^{-T} \tilde{B}_i L_i^{-1}) + \vech(\diag(u))\\
& \quad + D^W  \vech \{(\nu-r-1) W^{-T} - S^{-1} W\}  \\
&= \vech(\diag(u)) + D^W  \vech \big[ (n + \nu-r-1) W^{-T} -  S^{-1}W \\
& \quad + \sum_{i=1}^n (b_ib_i^T + \lambda_i a_i^T\Lambda_i + \Lambda_i a_i \lambda_i^T + \Lambda_i + L_i\tilde{B}_i L_i^T)W \big].
\end{aligned}
\end{equation*}
\end{proof}

\section{Proof of Lemma 4} 
In the second approach, we consider a Taylor expansion of 
\[
\log p(b_i|\theta_G, y_i) = y_i^T Z_i b_i - \sum_{j=1}^{n_i} h_{ij} (X_{ij}^T \beta + Z_{ij}^T b_i)- \frac{1}{2} b_i^T \Omega b_i + C_2,
\]
about $\hat{b}_i$, which satisfies 
\begin{equation} \label{bihat}
Z_i^T (y_i - g_i(X_i \beta + Z_i \hat{b}_i)) = \Omega \hat{b}_i,
\end{equation}
and 
\[
\lambda_i = \hat{b}_i, \quad \Lambda_i = \{Z_i^T H_i(X_i \beta+ Z_i \hat{b}_i) Z_i +  \Omega\}^{-1}.
\]
Differentiating \eqref{bihat} w.r.t. $\beta$ implicitly,
\begin{equation*}
\begin{aligned}
- Z_i^T H_i(X_i \beta + Z_i \hat{b}_i) (X_i \d \beta + Z_i \d \hat{b}_i &)= \Omega \d \hat{b}_i \\
- Z_i^T H_i(X_i \beta + Z_i \hat{b}_i) X_i \d \beta &= \Lambda_i^{-1} \d \hat{b}_i \\
\implies \nabla_\beta \lambda_i = \nabla_\beta \hat{b}_i &= - X_i^T H_i(X_i \beta + Z_i \hat{b}_i) Z_i \Lambda_i.
\end{aligned}
\end{equation*}
Differentiating $\Lambda_i$ w.r.t. $\beta$,
\begin{equation*}
\begin{aligned}
\d \Lambda_i &= - \Lambda_i Z_i^T \diag\{ h_i'''(X_i \beta + Z_i \hat{b}_i) \odot (X_i \d \beta+ Z_i \d \hat{b}_i) \} Z_i \Lambda_i \\
\d \Lambda_i &= - \Lambda_i Z_i^T \diag\{ h_i'''(X_i \beta + Z_i \hat{b}_i) \odot (X_i \d \beta + Z_i (\nabla_\beta \hat{b}_i)^T \d \beta) \} Z_i \Lambda_i \\
\d \vec(\Lambda_i) &= - (\Lambda_i Z_i^T \otimes \Lambda_i Z_i^T) E_{n_i}^T \vech[\diag\{ h_i'''(X_i \beta + Z_i \hat{b}_i) \odot (X_i + Z_i (\nabla_\beta \hat{b}_i)^T)\d \beta \}] \\
\nabla_\beta \vec(\Lambda_i) &= -U_i^T E_{n_i} (Z_i \Lambda_i \otimes Z_i \Lambda_i)
\end{aligned}
\end{equation*}
where $U_i$ is a $n_i(n_i+1)/2 \times p$ matrix where each row corresponds to an element of $\vech(I_{n_i})$ and only the rows corresponding to the diagonal elements are nonzero. The row corresponding to the $j$th diagonal element is given by $h_{ij}'''(X_{ij}^T \beta + Z_{ij}^T \hat{b}_i) (X_{ij}^T + Z_{ij}^T(\nabla_\beta \hat{b}_i)^T)$.

Let $\alpha_i = \tfrac{1}{2} h_i'''(X_i\beta + Z_i \hat{b}_i) \odot \diag\{ Z_i ( \Lambda_i + L_i \tilde{B}_i L_i^T ) Z_i^T \}$. Substituting these results in Theorem \ref{thm2}, 
\begin{equation*}
\begin{aligned}
\nabla_\beta \ell(\tiltheta) &= - \sum_{i=1}^n X_i^T H_i(X_i \beta + Z_i \hat{b}_i) Z_i \Lambda_i a_i + \sum_{i=1}^n X_i^T (y_i - g_i(\eta_i) ) \\
&\quad  - \tfrac{1}{2}\sum_{i=1}^n U_i^T E_{n_i}  (Z_i \Lambda_i \otimes Z_i \Lambda_i) \vec (\Lambda_i^{-1} + L_i^{-T} \tilde{B}_i L_i^{-1})
- \beta/ \sigma_\beta^2 \\
&= \sum_{i=1}^n \big[X_i^T \{y_i - g_i(\eta_i) - H_i(X_i \beta + Z_i \hat{b}_i) Z_i \Lambda_i a_i\} - \tfrac{1}{2} U_i^T \vech\{ Z_i ( \Lambda_i +   L_i \tilde{B}_i L_i^T ) Z_i^T \} \big] \\
& \quad - \beta/ \sigma_\beta^2 \\
&= \sum_{i=1}^n X_i^T \{y_i - g_i(\eta_i) - H_i(X_i \beta + Z_i \hat{b}_i) Z_i \Lambda_i a_i\} - \beta/ \sigma_\beta^2 \\
& \quad - \tfrac{1}{2} \sum_{i=1}^n  (X_i^T + \nabla_\beta \hat{b}_i Z_i^T) [h_i'''(X_i\beta + Z_i \hat{b}_i) \odot \diag\{ Z_i ( \Lambda_i + L_i \tilde{B}_i L_i^T ) Z_i^T \} ] \\
&= \sum_{i=1}^n X_i^T \{y_i - g_i(\eta_i) - H_i(X_i \beta + Z_i \hat{b}_i) Z_i \Lambda_i (a_i - Z_i^T \alpha_i) - \alpha_i \} - \beta/ \sigma_\beta^2,
\end{aligned}
\end{equation*}

Differentiating \eqref{bihat} w.r.t. $\omega$ implicitly,
\begin{equation*}
\begin{gathered}
- Z_i^T H_i(X_i \beta + Z_i \hat{b}_i) Z_i \d \hat{b}_i = \d\Omega \hat{b}_i + \Omega \d \hat{b}_i \\
\d \hat{b}_i = - \Lambda_i \d\Omega \hat{b}_i = - (\hat{b}_i^T \otimes \Lambda_i) \d\vec(\Omega) = - (\hat{b}_i^T \otimes \Lambda_i) 2 N_r (W \otimes I) E_r^T D^W \d\omega\\
\d \hat{b}_i= - (\hat{b}_i^T \otimes \Lambda_i + \Lambda_i\otimes \hat{b}_i^T ) (W \otimes I) E_r^T D^W \d\omega\\
\d \hat{b}_i= - ( \hat{b}_i^TW \otimes \Lambda_i + \Lambda_iW \otimes \hat{b}_i^T) E_r^T D^W \d\omega\\
\implies \nabla_\omega \lambda_i = \nabla_\omega \hat{b}_i = - D^W E_r (W^T \hat{b}_i \otimes \Lambda_i + W^T \Lambda_i \otimes \hat{b}_i) \\
\end{gathered}
\end{equation*}

Differentiating $\Lambda_i$ w.r.t. $\omega$,
\begin{equation*}
\begin{aligned}
\d \Lambda_i &= - \Lambda_i [Z_i^T \diag\{ h'''(X_i \beta + Z_i \hat{b}_i) \odot Z_i \d \hat{b}_i \} Z_i + \d\Omega] \Lambda_i \\
\d \Lambda_i &= - \Lambda_i Z_i^T \diag\{ h'''(X_i \beta + Z_i \hat{b}_i) \odot Z_i (\nabla_\omega \hat{b}_i)^T \d \omega \} Z_i \Lambda_i - \Lambda_i \d\Omega \Lambda_i\\
\d \vec(\Lambda_i) &= - (\Lambda_i Z_i^T \otimes \Lambda_i Z_i^T) E_{n_i}^T \d\vech[\diag\{ h'''(X_i \beta + Z_i \hat{b}_i) \odot Z_i (\nabla_\omega \hat{b}_i)^T\d \omega \}] \\
&\quad - (\Lambda_i \otimes \Lambda_i) 2 N_r (W \otimes I) E_r^T D^W \d\omega\\
\d \vec(\Lambda_i) &= - (\Lambda_i Z_i^T \otimes \Lambda_i Z_i^T) E_{n_i}^T \d\vech[\diag\{ h'''(X_i \beta + Z_i \hat{b}_i) \odot Z_i (\nabla_\omega \hat{b}_i)^T\d \omega \}] \\
&\quad - 2N_r (\Lambda_iW \otimes \Lambda_i) E_r^T D^W \d\omega\\
\nabla_\omega \vec(\Lambda_i) &= - V_i^T E_{n_i}  (Z_i \Lambda_i \otimes Z_i \Lambda_i) - 2D^W E_r(W^T \Lambda_i \otimes \Lambda_i)N_r 
\end{aligned}
\end{equation*}
where $V_i$ is a $n_i(n_i+1)/2 \times r$ matrix where each row corresponds to an element of $\vech(I_{n_i})$ and only the rows corresponding to the diagonal elements are nonzero. The row corresponding to the $j$th diagonal element is given by $h'''(X_{ij}^T \beta + Z_{ij}^T \hat{b}_i) Z_{ij}^T(\nabla_\omega \hat{b}_i)^T$.

Hence we have
\begin{equation*}
\begin{aligned}
&\tfrac{1}{2}\{\nabla_\omega \vec(\Lambda_i)\} \vec (\Lambda_i^{-1} + L_i^{-T} \tilde{B}_i L_i^{-1}) \\
& = \{ - \tfrac{1}{2} V_i^T E_{n_i}  (Z_i \Lambda_i \otimes Z_i \Lambda_i) - D^W E_r(W^T \Lambda_i \otimes \Lambda_i)N_r\} \vec (\Lambda_i^{-1} + L_i^{-T} \tilde{B}_i L_i^{-1})  \\ 
& = - \tfrac{1}{2} V_i^T \vech (Z_i (\Lambda_i + L_i \tilde{B}_i L_i^T) Z_i^T)  - D^W\vech \{(\Lambda_i + L_i \tilde{B}_i L_i^T )W \} \\
& =  - \tfrac{1}{2} \nabla_\omega \hat{b}_i Z_i^T  [h'''(X_i\beta + Z_i \hat{b}_i) \odot \diag\{ Z_i ( \Lambda_i + L_i \tilde{B}_i L_i^T ) Z_i^T \} ]  - D^W\vech \{(\Lambda_i + L_i \tilde{B}_i L_i^T )W \} \\
& = D^W E_r (W^T \hat{b}_i \otimes \Lambda_i + W^T \Lambda_i \otimes \hat{b}_i)  Z_i^T \alpha_i - D^W\vech \{(\Lambda_i + L_i \tilde{B}_i L_i^T )W \} \\
& = D^W\vech \{(\Lambda_i Z_i^T \alpha_i \hat{b}_i^T + \hat{b}_i \alpha_i^T Z_i \Lambda_i - \Lambda_i - L_i \tilde{B}_i L_i^T )W \}.
\end{aligned}
\end{equation*}

Substituting these results into Theorem \ref{thm2}, 
\begin{equation*}
\begin{aligned}
& \nabla_\omega \ell(\tiltheta) = \vech(\diag(u)) + D^W  \vech \{(\nu-r-1) W^{-T} - S^{-1} W\}  \\
& \quad + D^W \sum_{i=1}^n \vech \{(\Lambda_i Z_i^T \alpha_i \hat{b}_i^T + \hat{b}_i \alpha_i^T Z_i \Lambda_i - \Lambda_i - L_i \tilde{B}_i L_i^T )W \} \\
& \quad + D^W \sum_{i=1}^n \vech(W^{-T} - b_ib_i^T W) - D^W E_r \sum_{i=1}^n (W^T \hat{b}_i \otimes \Lambda_i + W^T \Lambda_i \otimes \hat{b}_i) a_i\\
&= \vech(\diag(u)) + D^W  \vech \{(n + \nu-r-1) W^{-T} - S^{-1} W\} \\
& \quad + D^W \sum_{i=1}^n \vech \{(\Lambda_i Z_i^T \alpha_i \hat{b}_i^T + \hat{b}_i \alpha_i^T Z_i \Lambda_i - \Lambda_i - L_i \tilde{B}_i L_i^T - b_ib_i^T- \Lambda_i a_i \hat{b}_i^T - \hat{b}_i a_i^T \Lambda_i)W \} \\
&= \vech(\diag(u)) + D^W \vech\{ (n + \nu-r-1) W^{-T} - S^{-1} W \\
&\quad  - \sum_{i=1}^n (\Lambda_i (a_i - Z_i^T \alpha_i ) \hat{b}_i^T + \hat{b}_i (a_i - Z_i^T \alpha_i )^T \Lambda_i + \Lambda_i + L_i \tilde{B}_i L_i^T + b_ib_i^T )W \}.
\end{aligned}
\end{equation*}

\end{document}